\documentclass[10pt,letterpaper]{article}

\usepackage[hang,perpage]{footmisc}

\usepackage{amsmath}
\usepackage{amsthm}
\usepackage{amssymb}
\usepackage{stmaryrd}
\usepackage{graphicx}
\usepackage{latexsym}
\usepackage{polynom}
\usepackage{times}
\usepackage[usenames]{color}
\usepackage{braket}
\usepackage{mathabx}
\usepackage{textcomp}
\usepackage{makeidx}
\usepackage{tocbibind}
\usepackage[bookmarks=true,backref=true,colorlinks=true,linkcolor=darkolivegreen,
citecolor=darkolivegreen,urlcolor=darkolivegreen]{hyperref}
\numberwithin{equation}{subsection}
\usepackage{authblk}

\theoremstyle{definition}
\newtheorem{ass}{Assumption}[section]
\newtheorem{theorem}[ass]{Theorem}
\newtheorem{prop}[ass]{Proposition}
\newtheorem{lemma}[ass]{Lemma}
\newtheorem{definition}[ass]{Definition}
\newtheorem{corollary}[ass]{Corollary}

\newtheorem{rem}[ass]{Remark}
\def\indexname{Index of terminology}
\makeindex

\newcommand{\captionfonts}{\footnotesize}
\makeatletter  
\long\def\@makecaption#1#2{%
  \vskip\abovecaptionskip
  \sbox\@tempboxa{{\captionfonts #1: #2}}%
  \ifdim \wd\@tempboxa >\hsize
    {\captionfonts #1: #2\par}
  \else
    \hbox to\hsize{\hfil\box\@tempboxa\hfil}%
  \fi
  \vskip\belowcaptionskip}
\makeatother   
\definecolor{darkolivegreen}{rgb}{0.333333, 0.419608, 0.1843140}
\allowdisplaybreaks

\newcommand\varpm{\mathbin{\vcenter{\hbox{%
  \oalign{\hfil$\scriptstyle+$\hfil\cr
          \noalign{\kern-.3ex}
          $\scriptscriptstyle({-})$\cr}%
}}}}
\newcommand\varmp{\mathbin{\vcenter{\hbox{%
  \oalign{$\scriptstyle{\phantom{\cdot}-}$\cr
          \noalign{\kern-.3ex}
          \hfil$\scriptscriptstyle({+})$\hfil\cr}%
}}}}

\makeatletter
\def\printnotation{{%
\def\indexname{Index of notation}
\begin{theindex}
\@input{\jobname.ntn}
\end{theindex}
}}
\makeatother

\makeglossary

\begin{document}



\title{Stability study of a model for the Klein-Gordon equation  
in Kerr space-time~II}

\author[1,2,3]{Horst Reinhard Beyer}
\author[3]{Miguel Alcubierre}
\author[4]{Miguel Megevand}
\affil[1]{Instituto Tecnol\'ogico Superior de Uruapan, 
Carr. Uruapan-Carapan No. 5555, Col. La Basilia, Uruapan, 
Michoac\'an, M\'exico}
\affil[2]{Goethe Universit\"{a}t Frankfurt am Main, Institut f\"{u}r Theoretische Physik,  Max-von-Laue-Strasse  1,  60438 Frankfurt, Germany}
\affil[3]{Instituto de Ciencias Nucleares, Universidad Nacional
  Aut\'onoma de M\'exico, Circuito Exterior C.U., A.P. 70-543,
  M\'exico D.F. 04510, M\'exico}
\affil[4]{Instituto de F\'isica Enrique Gaviola, CONICET. Ciudad Universitaria, 5000 C\'ordoba, Argentina.}
\date{\today}                                     

\maketitle

\begin{abstract}
The present paper is a follow-up of our previous paper that derives a slightly simplified model equation
for the Klein-Gordon equation, describing the propagation of a scalar field of mass $\mu$
in the background of a rotating black hole and, among others, supports 
the instability of the field down to 
$a/M \approx 0.97$. The latter result was derived numerically. This paper
gives corresponding rigorous results, supporting instability of the field down to 
$a/M \approx 0.979796$.
\end{abstract}

\section{Introduction}
\label{introduction}

The question of the stability of the solutions of the 
Klein-Gordon equation, describing a massive
scalar field inside the gravitational field of rotating (Kerr-) black hole, (in Boyer-Lindquist coordinates,) is an
important model problem, in view of the stability of the Kerr metric.  
Results of Damour, Deruelle and Ruffini (\cite{damour}, 1976), of 
Zouros and Eardley (\cite{zouros}, 1979) and Detweiler (\cite{detweiler}, 1980)  
indicate the existence of unstable modes
for `small' masses of the black hole. This is a surprising 
result because the Klein-Gordon equation, describing 
a massive field, is a perturbation of the
wave equation on Kerr background, by a positive potential. Intuitively, it might be expected that such perturbations stabilize solutions, and the solutions of the latter equation are indeed stable
\cite{whiting,finster,dafermos,shlapen2}. The rigorous 
proof of the results above was of considerable interest, since, if true,  
this could indicate that infalling matter destabilizes 
Kerr black holes. Beyer (\cite{beyer4}, 2011) proves that the restrictions of 
the solutions of the separated, in the azimuthal coordinate, Klein-Gordon 
field are stable for sufficiently large masses $\mu > 0$ of the 
field. 
\begin{equation} \label{massineq}
\mu \geqslant \frac{|m|a}{2Mr_{+}} 
\sqrt{1 + \frac{2M}{r_{+}}}
\, \, .
\end{equation}
Here $M > 0$ is the mass 
of the black hole, 
$0 \leqslant a < M$ is the rotational parameter, 
$m \in {\mathbb{Z}}$ is the 
`azimuthal separation parameter' and 
\begin{equation*}
r_{+} := M + 
\sqrt{M^2 - a^2} \, \, .
\end{equation*}
The result is consistent with \cite{damour}, 
but contradicts results of \cite{zouros}. A numerical investigation by Furuhashi and Nambu (\cite{furuhashi}, 2004) finds unstable modes for $\mu M \sim 1$ and 
$(a/M) = 0.98$. A numerical investigation by Cardoso et al. (\cite{cardoso}, 2004), finds unstable 
modes for $\mu M \leqslant 1$ and $0.98 \leqslant (a/M) < 1$.
A numerical investigation by Strafuss and Khanna (\cite{khanna}, 2005)
finds unstable 
modes for $\mu M \sim 1$ and $(a/M) = 0.9999$.
A numerical investigation by Konoplya and Zhidenko (\cite{konoplya}, 2006) confirms 
the results of Beyer (\cite{beyer1}, 2001, \cite{beyer4}, 2011). In addition, no 
unstable modes are found
also for $\mu M \ll 1$ 
and $\mu M \sim 1$. ($0 \leqslant a \leqslant 
0.995)$. The latter result contradicts, 
in particular, analytical results from Detweiler \cite{detweiler}.
An analytical study by Hod and Hod (\cite{hod}, 2010) finds
unstable modes for $\mu M \sim 1$ with a growth rate,
\begin{equation*} 
1.7 \cdot 10^{-3} M^{-1} \, \, , 
\end{equation*}
which is four orders of 
magnitude larger than previously estimated. 
\newline
\linebreak
There was a mounting evidence that the solutions of the Klein-Gordon equation on a Kerr background are unstable, if the estimate (\ref{massineq}) is violated.
By negelecting ``small" terms, in the sense of the used operator-theoretic methods, 
Beyer, Alcubierre \& Megevand, (\cite{beyer5}, 2013) create a spherically symmetric model equation 
that is closely related to the Klein-Gordon equation on a Kerr background and whose modes can 
be expressed in terms of Coulomb wave functions.
Analogous to the Klein-Gordon equation on a Kerr background,
the model equation is of the form 
\begin{equation} \label{abstractwaveequation}
(u^{\prime})^{\prime}(t) + i B u^{\prime}(t) + A  u(t) = 0 \, \, , 
\end{equation} 
for every $t \in {\mathbb{R}}$, where $u$ is the unknown function, assuming values in a weighted 
$L^2$-space $X$,  
$A$ is a densely-defined, linear and self-adjoint operators in $X$
and $B$ is a bounded linear and self-adjoint operator in $X$. The operators $A$ and $B$
do not commute, as is the case also for the Klein-Gordon equation on a Kerr background.
\footnote{The spectral parameter $\lambda$ is a kind of "frequency." For this, 
we note that if $u(t) = e^{i \lambda t} \xi$, for every $t \in {\mathbb{R}}$ and 
$\xi$ is an element of the domain of $A$, then (\ref{abstractwaveequation}) would lead to the
equation $(A - \lambda B - \lambda^2)\xi = 0$.}
The stability of the solutions of (\ref{abstractwaveequation}) is governed by the spectrum of 
the corresponding operator polynomial 
\begin{equation} \label{operatorpolynomial}
A - \lambda B - \lambda^2
\end{equation}
where $\lambda \in {\mathbb{C}}$, i.e., by those $\lambda \in {\mathbb{C}}$, 
for which the operator in (\ref{operatorpolynomial}) is not bijective. 
The solutions of the model equation are unstable down to rotational parameters 
$a/M \approx 0.9718$.
Subsequently, Shlapentokh-Rothman, (\cite{shlapen}, 2014) proved the instability of the solutions of the 
Klein-Gordon equation, describing a massive
scalar field on a Kerr background, in the following sense. For each choice of $m \in {\mathbb{Z}}^{*}$, there is a countable family of intervals of
masses $\mu$ associated to exponentially growing solutions (indexed by $l \in \{|m|, |m| + 1,\dots \}$). These intervals
have an accumulation point at
\begin{equation} \label{massbound}
\frac{|m|a}{2Mr_{+}} \, \, .
\end{equation}
In addition, these unstable modes 
exhibit superradiance, i.e., the corresponding frequency $\omega$ satisfies the inequality
\begin{equation*}
m a \, {\textrm{Re}}(\omega) -  2Mr_{+} |\omega|^2 > 0 \, \, .
\end{equation*}
Still, there is very much an implicit dependence on $a/M$, as in the numerical results.
The precise dependence of the instability on the parameters, including 
the value of $a/M$ triggering the   
onset of the instability is not yet clear. 
\newline
\linebreak 
The present paper is a follow-up of our paper \cite{beyer5}. It 
continues the study of our model problem, with the purpose of shedding some light on 
the dependence of the instability on the parameters, including 
the value of $a/M$ triggering the   
onset of the instability. Here, it needs to be taken into account 
that \cite{beyer5} 
reduces the finding of 
unstable modes of (\ref{abstractwaveequation}) to the finding of the 
solutions of a quartic inside the subset  
\begin{figure} 
\centering
\includegraphics[width=5.6cm,height=5.6cm]{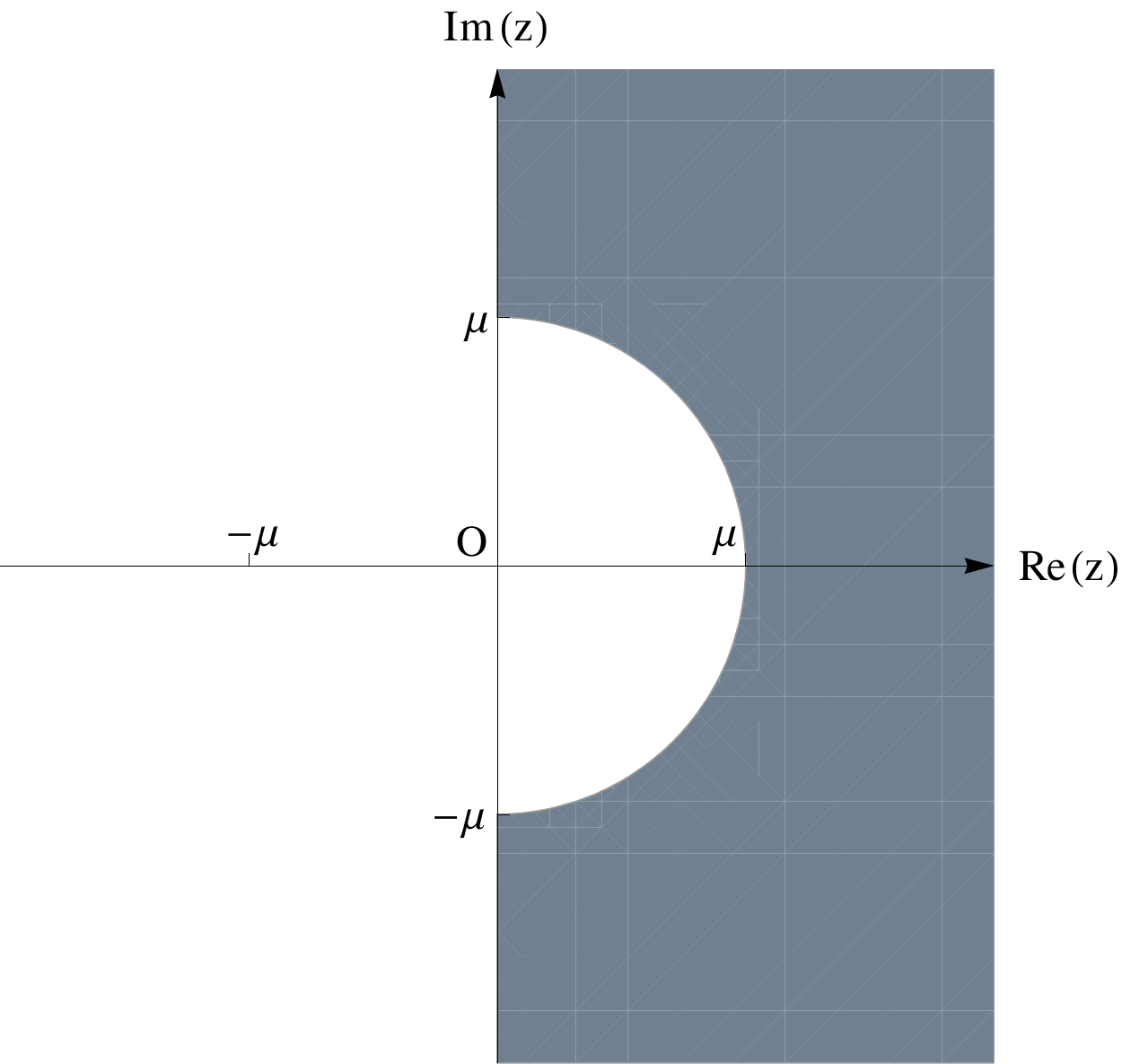}
\caption{Sketch of the subset $\Omega_1$ of the complex plane.}
\label{fig1}
\end{figure}

\begin{equation} \label{subset}
\Omega_1 := ({\mathbb{C}} \setminus B_{\mu}(0)) 
\cap ((0,\infty) \times {\mathbb{R}})
\end{equation}
of the complex plane, but finds these
solutions only numerically. Here, ${\mathbb{C}}$ denotes the field of complex numbers, 
$B_{\mu}(0)$ denotes the closed ball of radius $\mu$ around $0$ and 
$(0,\infty) \times {\mathbb{R}}$ denotes the open right half-plane.
The present paper focuses on obtaining analytical 
information on these solutions. For the study of the model problem, we 
assume throughout that
\begin{equation*}
M > 0 \, \, , \, \, 0 < a < M \, \, , \, \, \mu >  0 \, \, , \, \, m \in {\mathbb{Z}}
\, \, , \, \, l \in \{|m|, |m| + 1, \dots\} \, \, .
\end{equation*}
The following is Lemma~3.19 in \cite{beyer5}, reducing the finding of 
unstable modes of (\ref{abstractwaveequation}) to the solution of a quartic 
equation and providing the starting point of the investigation.
\begin{lemma}
If $R = r_{+}$, i.e., $R_{-} = 2 (M^2 - a^2)^{1/2}$, 
and $\lambda$ satisfies the condition   
\begin{equation} \label{exceptions1}
\lambda \neq 
- \, \frac{1}{2 M r_{+}} \, [ma + i k (M^2 - a^2)^{1/2}]
\, \, , 
\end{equation}
where $k \in {\mathbb{Z}}$,
then $\lambda \in {\mathbb{R}} \times (-\infty,0)$ is such 
that   
$\ker(A - \lambda B - \lambda^2)$ is non-trivial if and only if  
\begin{equation*}
\lambda = - \, \frac{i}{2z} \, (z^2 - \mu^2) \, \, , 
\end{equation*}
where $z \in \Omega_1$ satisfies
\begin{align}
z^4 & + \frac{(2n + 1)(M^2 - a^2)^{1/2} + i ma}{M(2r_{+}+M)}
\, z^3 + \frac{c_{l}- 2 M^2 \mu^2}{M(2r_{+}+M)} \, z^2  \nonumber \\
& + \frac{(2n + 1)(M^2 - a^2)^{1/2} + i ma}{M(2r_{+}+M)} 
\, \mu^2 z + \frac{M}{2r_{+}+M} \, \mu^4 = 0 \, \, ,  \label{spectralcondition}
\end{align}
for some $n \in {\mathbb{N}}$ and where $c_l := l (l+1)$.
\end{lemma}
We note that if $z \in \Omega_1$ satisfies 
(\ref{spectralcondition}), then $z^{*}$ is a solution 
of (\ref{spectralcondition}), where $m$ is replaced by 
$-m$, that is contained in 
$\Omega_1$. Further, we note that the 
coefficients of the first and third power of $z$ of the quartic 
(\ref{spectralcondition}) 
are neither real nor 
purely imaginary, if $m \neq 0$. The remaining coefficients are real. For the model problem, there
is a stability condition given by (\ref{massineqmodel})
from 
Corollary~3.16 of \cite{beyer5}:
\begin{corollary} \label{lowerbound}
If $R = r_{+}$, i.e., $R_{-} = 2 (M^2 - a^2)^{1/2}$, 
\begin{equation} \label{massineqmodel}
 \mu^2 \leq \frac{l(l+1)}{2 M (r_{+} + 2 M)} \, \, ,
\end{equation}
and $\lambda \in {\mathbb{C}} \setminus {\mathbb{R}}$ satisfies the condition  (\ref{exceptions1}),
then 
$\ker(A - \lambda B - \lambda^2)$ is trivial. 
\end{corollary}
As a consequence, for the case $m=0$, $\ker(A - \lambda B - \lambda^2)$ is trivial.
Hence, in the following, we assume throughout that
\begin{equation*}
m \in {\mathbb{Z}}^{*} \, \, .
\end{equation*}
In following sections, we proceed to show the existence of solutions 
of \eqref{spectralcondition} inside $\Omega_1$, and hence
the existence of unstable modes in a subregion of the parameter
space.
Although the roots of any forth degree polynomial, such as that
in~\eqref{spectralcondition}, are known explicitly, the 
 expressions, seen as functions of all the parameters involved
(i.e. $a$, $\mu$, $n$, $l$ and $m$), are too complicated to get
any
intuitive understanding of the problem just from their analytical form. 
Thus, we are using other analytical methods for determining the location of the
roots, in particular, {\it Routh-Hurwitz criteria}, for the localization of roots in {\it half-planes}, the {\it Schur-Cohn algorithm}, for the localization of roots inside the {\it closed
unit disk}, {\it Rouche's theorem}, for the localization of roots in {\it general 
domains}, calculation of discriminants of polynomials and direct estimates. We give $2$
approaches. Approach~1 shows the existence of roots in (\ref{subset}), 
for sufficiently large $a/M$, without giving a lower bound for $a/M$.
Approach~2 shows the existence of roots in (\ref{subset}),
for $a/M$ satisfying the inequality
(\ref{newinstabilitycondition2}).

\section{Approach~1}
In the following, we use conformal transformations to transform $\Omega_1$
into a subset ($\Omega_2$, see (\ref{Omega2})) of the complex plane that is suitable,
for the application, in particular, of the Schur-Cohn algorithm. 
We note by $P$ the polynomial in 
(\ref{spectralcondition}), i.e., 
\begin{align*}
P(z) := z^4 & + \frac{(2n + 1)(M^2 - a^2)^{1/2} + i ma}{M(2r_{+}+M)}
\, z^3 + \frac{c_{l}- 2 M^2 \mu^2}{M(2r_{+}+M)} \, z^2  \nonumber \\
& + \frac{(2n + 1)(M^2 - a^2)^{1/2} + i ma}{M(2r_{+}+M)} 
\, \mu^2 z + \frac{M}{2r_{+}+M} \, \mu^4  \, \, , 
\end{align*}
for every $z \in {\mathbb{C}}$.
Then, 
\begin{align*} 
& \frac{P(z)}{\mu^4} = \left(i \, \frac{z}{\mu}\right)^4 - \frac{ma - i (2n + 1)(M^2 - a^2)^{1/2}}{\mu M(2r_{+}+M)}
\, \left(i \, \frac{z}{\mu}\right)^3 \nonumber \\
& + \frac{2 M^2 \mu^2  - l(l+1)}{\mu^2M(2r_{+}+M)} \, \left(i \, \frac{z}{\mu}\right)^2  \\
& + \frac{ma - i (2n + 1)(M^2 - a^2)^{1/2}}{\mu M(2r_{+}+M)} 
\, \left(i \, \frac{z}{\mu}\right) + \frac{M}{2r_{+}+M} \, \, ,   \nonumber 
\end{align*}
 and hence 
\begin{equation*}
P(z) = \mu^4 f\left(i \, \frac{z}{\mu}\right) \, \, , 
\end{equation*}
for every $z \in {\mathbb{C}}$, where 
\begin{definition}{(\bf Definition of $f$)}
We define $f : {\mathbb{C}} \rightarrow {\mathbb{C}}$ by 
\begin{align} \label{definitionoff}
f(u) &:= u^4 - \frac{ma - i (2n + 1)(M^2 - a^2)^{1/2}}{\mu M(2r_{+}+M)}
\, u^3 + \frac{2 M^2 \mu^2  - l(l+1)}{\mu^2M(2r_{+}+M)} \, u^2  \nonumber \\
& \quad \, \, \, + \frac{ma - i (2n + 1)(M^2 - a^2)^{1/2}}{\mu M(2r_{+}+M)} 
\, u + \frac{M}{2r_{+}+M}  \, \, ,  
\end{align}
for every $u \in {\mathbb{C}}$. 
\end{definition}

\begin{figure} 
\centering
\includegraphics[width=9.52380952380952380952cm,height=5.6cm]{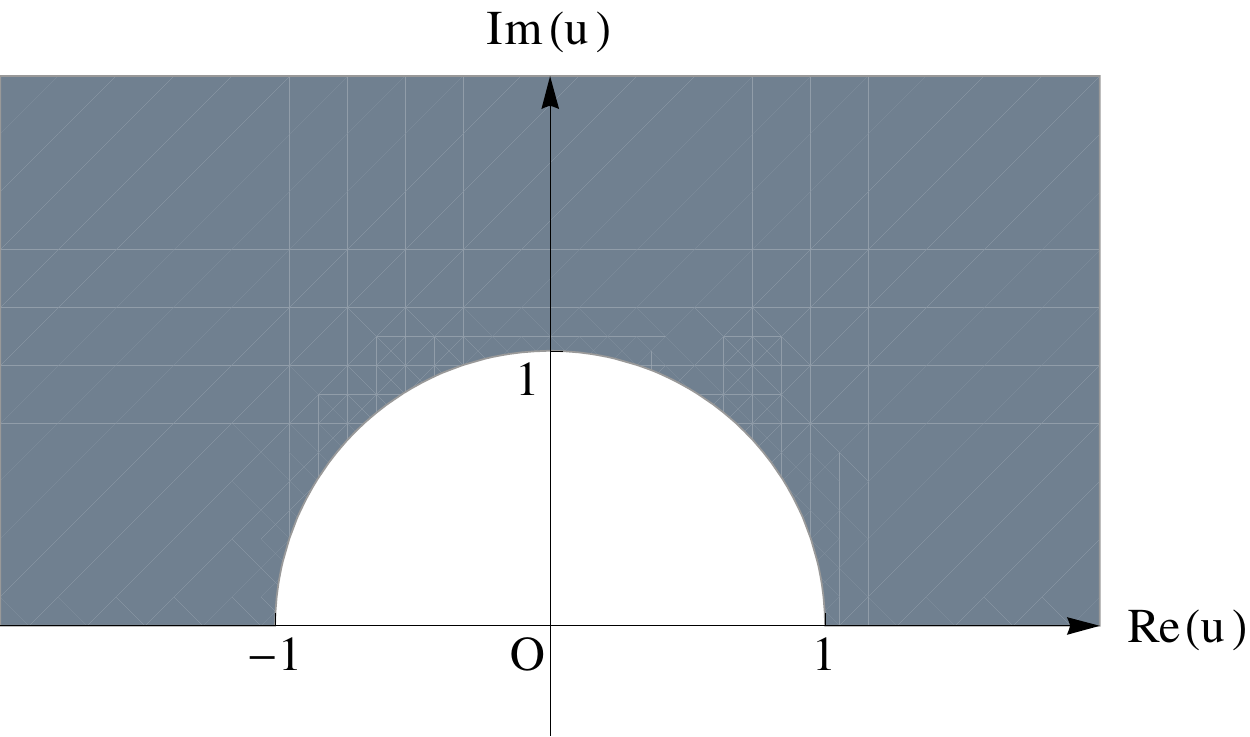}
\caption{The domain of values of $u$ inside the domain of $f$, leading on unstable $\lambda$ is given by 
the complement of the upper closed half-disk, shaded in gray.}.
\label{fig2}
\end{figure}

As a consequence, we arrive at the following:

\begin{lemma}{(\bf Instability in terms of roots of $f$)}
If $R = r_{+}$, i.e., $R_{-} = 2 (M^2 - a^2)^{1/2}$, 
and $\lambda$ satisfies (\ref{exceptions1}),
then $\lambda \in {\mathbb{R}} \times (-\infty,0)$ is such 
that   
$\ker(A - \lambda B - \lambda^2)$ is non-trivial if and only if  
\begin{equation*}
\lambda = - \frac{\mu}{2u } \, (1 + u ^2) = - \frac{\mu}{2} \left( u + \frac{1}{u} \right) \, \, ,
\end{equation*}
for a root $u$ of $f$ contained in $({\mathbb{C}} \setminus B_{1}(0)) \cap ({\mathbb{R}} \times (0,\infty))$. 
\end{lemma}

We note that 
\begin{align*}
\left( U_{1}(0) \cap ({\mathbb{R}} \times (-\infty,0))
\rightarrow ({\mathbb{C}} \setminus B_{1}(0)) \cap ({\mathbb{R}} \times (0,\infty))
, 
u \mapsto \frac{1}{u}
\, \right)
\end{align*}
is biholomorphic. Further, for $w \in U_1(0) \cap ({\mathbb{R}} \times (-\infty,0))$,
it follows that 
\begin{align*}
f(w^{-1}) & = w^{-4} - \frac{ma - i (2n + 1)(M^2 - a^2)^{1/2}}{\mu M(2r_{+}+M)}
\, w^{-3} + \frac{2 M^2 \mu^2  - l(l+1)}{\mu^2M(2r_{+}+M)} \, w^{-2}  \\
& \quad \, \, \, + \frac{ma - i (2n + 1)(M^2 - a^2)^{1/2}}{\mu M(2r_{+}+M)} 
\, w^{-1} + \frac{M}{2r_{+}+M} \\
&= w^{-4} 
\left[ 
1 - \frac{ma - i (2n + 1)(M^2 - a^2)^{1/2}}{\mu M(2r_{+}+M)}
\, w + \frac{2 M^2 \mu^2  - l(l+1)}{\mu^2M(2r_{+}+M)} \, w^{2}  \right. \\
& \left. \qquad \quad \, \, \, + \frac{ma - i (2n + 1)(M^2 - a^2)^{1/2}}{\mu M(2r_{+}+M)} 
\, w^{3} + \frac{M}{2r_{+}+M} \, w^4
\right]  \, \, .
\end{align*}

As a consequence, we define the following.

\begin{definition}{(\bf Definition of $p, p_{e}$ and $\delta$)}
We define, 
\begin{align}  \label{definitionofp}
p(w) &:= w^4 + \frac{ma - i (2n + 1)(M^2 - a^2)^{1/2}}{\mu M^2} 
\, w^{3} \\
& \quad \, \, \, + \frac{2 M^2 \mu^2  - l(l+1)}{\mu^2 M^2} \, w^{2}  
- \frac{ma - i (2n + 1)(M^2 - a^2)^{1/2}}{\mu M^2}
\, w + \frac{2r_{+}+M}{M} 
\, \, ,  \nonumber \\
& \, \, = w^4 + \frac{\frac{ma}{M} - i (2n + 1)\sqrt{1 - \frac{a^2}{M^2}}}{\mu M} 
\, w (w^{2} - 1) \nonumber \\
& \quad \, \, \, + \left[ 2 - 
\frac{l(l+1)}{\mu^2 M^2} \right] w^{2}  
+ 3 + 2 \sqrt{1 - \frac{a^2}{M^2}} \, \, , \nonumber \\
& \, \, = w^4 + \frac{m - m \left(1 - \frac{a}{M}\right) - i (2n + 1)\sqrt{1 - \frac{a^2}{M^2}}}{\mu M} 
\, w (w^{2} - 1) \nonumber \\
& \quad \, \, \, + \left[ 2 - 
\frac{l(l+1)}{\mu^2 M^2} \right] w^{2}  
+ 3 + 2 \sqrt{1 - \frac{a^2}{M^2}} \, \, , \nonumber \\
& \, \, = w^4 + \frac{m}{\mu M} 
\, w (w^{2} - 1) \nonumber + \left[ 2 - 
\frac{l(l+1)}{\mu^2 M^2} \right] w^{2}  
+ 3 \\
& \quad \, \, \, \, - \frac{m \left(1 - \frac{a}{M}\right) + i (2n + 1)\sqrt{1 - \frac{a^2}{M^2}}}{\mu M} 
\, w (w^{2} - 1)  + 2 \sqrt{1 - \frac{a^2}{M^2}} \, \, , \nonumber \\
& = 
p_{e}(w) + \delta(w)
\, \, , \nonumber
\end{align}
for every $w \in {\mathbb{C}}$, where 
\begin{align} \label{defofpeanddelta}
p_{e}(w) & := w^4 + \frac{m}{\mu M} 
\, w (w^{2} - 1) + \left[ 2 - 
\frac{l(l+1)}{\mu^2 M^2} \right] w^{2}  
+ 3 \, \, , \\
\delta(w) & :=  - \frac{m \left(1 - \frac{a}{M}\right)}{\mu M} 
\, w (w^{2} - 1) + 2 \, \sqrt{1 - \frac{a^2}{M^2}} \left[1 -i \, \frac{2n+1}{2 \mu M} \, 
\, w (w^{2} - 1)\right] \, \, , \nonumber 
\end{align}
for every $w \in {\mathbb{C}}$.
\end{definition}

Hence, we arrive at the following 

\begin{figure} 
\centering
\includegraphics[width=9.52380952380952380952cm,height=4.57142857142857142857cm]{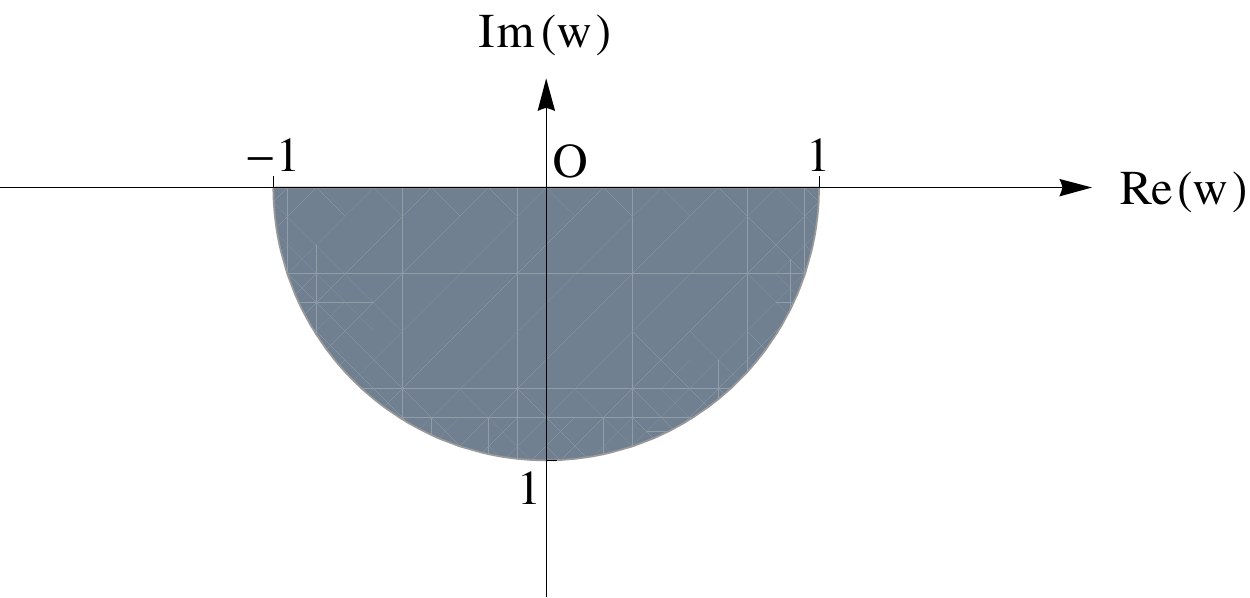}
\caption{The domain of values of $w$ inside the domain of 
$p$, leading on unstable $\lambda$ is given by the lower open half disk, shaded in gray.}.
\label{fig3}
\end{figure}

\begin{lemma}{(\bf Instability in terms of roots of $p$)}
If $R = r_{+}$, i.e., $R_{-} = 2 (M^2 - a^2)^{1/2}$, 
and $\lambda$ satisfies (\ref{exceptions1}),
then $\lambda \in {\mathbb{R}} \times (-\infty,0)$ is such 
that   
$\ker(A - \lambda B - \lambda^2)$ is non-trivial, if and only if  
\begin{equation*}
\lambda = - \frac{\mu}{2w} \, (1 + w^2) = - \frac{\mu}{2} \left(w + \frac{1}{w}
\right) \, \, , 
\end{equation*}
for a root $w$ of $p$ contained in 
\begin{equation} \label{Omega2}
\Omega_2 := U_{1}(0) \cap ({\mathbb{R}} \times (-\infty,0)) \, \, .
\end{equation}
\end{lemma}

We note that 
\begin{equation*}
p_{e}(w) =  w^4 + {\alpha} w^3 + ({\beta}-4) w^2 - {\alpha} w + 3 \, \, , 
\end{equation*}
where 
\begin{equation*}
\alpha = \frac{m}{\mu M} \, \, , \, \, 
\beta = 6 - 
\frac{l(l+1)}{\mu^2 M^2} \, \, ,
\end{equation*}
are dimensionless.
Also, we note that, since 
\begin{equation*}
\frac{l(l+1)}{6 M^2} \geqslant \frac{l(l+1)}{2M (3M + \sqrt{M^2 - a^2} \, )} 
= \frac{l(l+1)}{2 M (r_{+} + 2M)}  
\, \, ,
\end{equation*}
the condition that 
\begin{equation*}
\mu^2 > \frac{l(l+1)}{6 M^2} \, \, ,
\end{equation*}
implies that 
\begin{equation*}
\mu^2 > \frac{l(l+1)}{2 M (r_{+} + 2M)} \, \, .
\end{equation*}
More generally, in the following, we define for $\alpha, \beta \in {\mathbb{R}}$ the 
polynomial $q$ by 
\begin{equation} \label{qalphabeta}
q_{\alpha,\beta}(w) :=  w^4 + {\alpha} w^3 + ({\beta}-4) w^2 - {\alpha} w + 3  \, \, , 
\end{equation}
for every $w \in {\mathbb{C}}$.\footnote{As a side remark, that it turns out that calling 
the coefficient  of $w^2$ in (\ref{qalphabeta}) ``$\beta - 4$", instead of ``$\beta$," is going to
simplify calculations in future, for some unknown reason.}
As a consequence, if 
\begin{equation*}
\alpha = \frac{m}{\mu M} \, \, , \, \, 
\beta = 6 - 
\frac{l(l+1)}{\mu^2 M^2} \, \, , 
\end{equation*}
then 
\begin{equation*}
p_{e} = q_{\alpha,\beta} \, \, .
\end{equation*}
In the following, we are going to apply the Cohn-Schur algorithm to find the number 
of roots of $q_{(\alpha,\beta)}$ in the open ball of radius $1$ around the origin,  $U_{1}(0)$, of the complex plane.
\begin{theorem}{(\bf Number of roots of $q_{\alpha,\beta}$ inside $U_{1}(0)$)}
\label{theorem1}
Let $\alpha, \beta \in {\mathbb{R}}$ be such that 
\begin{equation*}
 4 < \alpha^2 \, \, , \, \, 0 < \beta \, \, .
\end{equation*} 
Then $q_{\alpha,\beta}$ has $2$ roots inside $U_{1}(0)$, 
where multiple roots are counted with their multiplicity.
\end{theorem}

\begin{proof}
In the following, we are going to apply Theorem~6.8c of \cite{henrici}.
Here $T^{k}q_{\alpha,\beta}$, $k=1,\dots,4$, denote the iterated Schur transforms of $q_{\alpha,\beta}$, 
a $*$ indicates a reciprocal polynomial and $\gamma_k := (T^{k}q_{\alpha,\beta})(0)$, 
for $k=1,\dots,4$. It follows for every $w \in {\mathbb{C}}$
that
\begin{align*}
q_{\alpha,\beta}(w) & = w^4 + {\alpha} w^3 + ({\beta}-4) w^2 - {\alpha} w + 3 \, \, , \\
q_{\alpha,\beta}^{*}(w) & = 3 w^4 - {\alpha} w^3 + ({\beta}-4) w^2 + {\alpha} w + 1 \, \, , \\
(Tq_{\alpha,\beta})(w) & = 3 q_{\alpha,\beta}(w) - q_{\alpha,\beta}^{*}(w) \\
& = 3 w^4 + 3 {\alpha} w^3 + 3({\beta}-4) w^2 - 3a w + 9 \\
& \quad \, -
[3 w^4 - {\alpha} w^3 + ({\beta}-4) w^2 + {\alpha} w + 1] \\
& = 4 {\alpha} w^3 + 2({\beta}-4) w^2 - 4 {\alpha} w + 8 \\
(Tq_{\alpha,\beta})^{*}(w) & = 8 w^3 - 4 {\alpha} w^2 + 2({\beta}-4) w + 4 \alpha \, \, \\
(T^2q_{\alpha,\beta})(w) & = 8 (Tq_{\alpha,\beta})(w) - 4 \alpha (Tq_{\alpha,\beta})^{*}(w) \\
& = 8 [4 {\alpha} w^3 + 2({\beta}-4) w^2 - 4 {\alpha} w + 8] \\
& \quad - 4 {\alpha} [8 w^3 - 4 {\alpha} w^2 + 2({\beta}-4) w + 4 \alpha] \\
& = 16({\alpha}^2 + {\beta} - 4) w^2 - 8 {\alpha} [4 + ({\beta}-4)] w + 16 (4 - {\alpha}^2) \\
& = 16({\alpha}^2 + {\beta} - 4) w^2 - 8 {\alpha} {\beta} w + 16 (4 - {\alpha}^2) \, \, , \\ 
(T^2q_{\alpha,\beta})^{*}(w) & = 16 (4 - {\alpha}^2) w^2 - 8 {\alpha} {\beta} w + 16({\alpha}^2 + {\beta} - 4) \, \, , \\
(T^3q_{\alpha,\beta})(w) & = 16 (4 - {\alpha}^2) [16({\alpha}^2 + {\beta} - 4) w^2 - 8 {\alpha} {\beta} w + 16 (4 - {\alpha}^2)] \\
& \quad \, 
- 16({\alpha}^2 + {\beta} - 4) [16 (4 - {\alpha}^2) w^2 - 8 {\alpha} {\beta} w + 16({\alpha}^2 + {\beta} - 4)] \\
& = 128 {\alpha} {\beta} [2({\alpha}^2-4) + {\beta}] w + 256 [(4 -{\alpha}^2)^2 - ({\alpha}^2 - 4 + {\beta})^2] \\
& = 128 {\alpha} {\beta} [2({\alpha}^2-4) + {\beta}] w + 256 {\beta} [2(4 - {\alpha}^2) - {\beta}] \\
& = 128 {\beta} [2({\alpha}^2 - 4) + {\beta}] (\alpha w  - 2) \, \, , \\
(T^3q_{\alpha,\beta})^{*}(w) & = 128 {\beta} [2({\alpha}^2 - 4) + {\beta}] (-2 w  + \alpha) \, \, , \\
(T^4q_{\alpha,\beta})(w) & = - 256 {\beta} [2({\alpha}^2 - 4) + {\beta}] \cdot 128 {\beta} [2({\alpha}^2 - 4) + {\beta}] (\alpha w  - 2) \\
&
\quad \, - 128 {\beta} {\alpha} [2({\alpha}^2 - 4) + {\beta}] \cdot 128 {\beta} [2({\alpha}^2 - 4) + {\beta}] (-2 w  + \alpha) \\
& = - 128 \cdot 256 {\beta}^2 [2({\alpha}^2 - 4) + {\beta}]^2 (\alpha w  - 2) \\
&
\quad \, - 128^2 {\beta}^2 \alpha [2({\alpha}^2 - 4) + {\beta}]^2 (-2 w  + \alpha) \\
& = 256^2 {\beta}^2 [2({\alpha}^2 - 4) + {\beta}]^2 - 128^2 {\beta}^2 \alpha^2 [2({\alpha}^2 - 4) + {\beta}]^2 \\
& = 128^2 (4 - \alpha^2) {\beta}^2 [2({\alpha}^2 - 4) + {\beta}]^2 \, \, , \\
\gamma_1 & = 8 \, \, , \\
\gamma_2 & = - 16 ({\alpha}^2 -4)  \, \, , \\
\gamma_3 & =  - 2 \cdot 128  {\beta} [2({\alpha}^2 - 4) + {\beta}] \, \, , \\
\gamma_4 & = - 128^2 ({\alpha}^2 - 4)
{\beta}^2 [2({\alpha}^2 - 4) + {\beta}]^2 \, \, .
\end{align*}
As a consequence, we conclude that the conditions
\begin{equation*}
{\alpha}^2 > 4 \, \, \wedge \, \, {\beta} > 0
\end{equation*}
imply that 
\begin{equation*}
\gamma_1 > 0 \, \, , \, \, \gamma_2 < 0 \, \, , \, \, \gamma_3 < 0 \, \, , \, \, \gamma_4 < 0
\end{equation*}
and hence that the corresponding indices $k_1,k_2,k_3$ are given by 
\begin{equation*}
k_1 = 2 \, \, , \, \, k_2 = 3 \, \, , \, \,
k_3 = 4 \, \, .
\end{equation*}
Therefore, according to Theorem~6.8c of \cite{henrici}, {\it the number of roots of $q_{\alpha,\beta}$ in $U_{1}(0)$, multiple roots counted with their multiplicity, is given by} 
\begin{align*}
\sum_{j=1}^{3} (-1)^{j-1}(4 + 1 - k_{j}) & =
5 - k_1 - (5 - k_2) + 5 - k_3 =  5 - 2 - (5 - 3) + 5 - 4 \\
& =  3 - 2 + 1 = 2 \, \, .
\end{align*}
\end{proof}
In the next step, we calculate the discriminant of the polynomial $q_{\alpha,\beta}$, to 
obtain information on the multiplicities of the roots of $q_{\alpha,\beta}$.
\begin{theorem}{(\bf Calculation of the discriminant of $q_{\alpha,\beta}$)}
\label{theorem2}
Let $\alpha, \beta \in {\mathbb{R}}$ be such that 
\begin{equation*}
4 < \alpha^2 < 6 \, \, , \, \, \frac{4}{100} < \beta < \frac{165}{100}  
\, \, .
\end{equation*}
Then $q_{\alpha,\beta}$ 
\begin{itemize}
\item[(i)] has $4$ pairwise different roots, 
\item[(ii)] $2$ of these roots are real,
\item[(iii)] and $2$ of these roots are non-real and conjugate 
complex.
\end{itemize}
\end{theorem}

\begin{proof}
From direct calculation, it follows that the discriminant $\triangle$ of $q_{\alpha,\beta}$ is
given by 
\begin{align*}
\triangle & = 4 {\alpha}^6+{\alpha}^4 {\beta}^2-80 {\alpha}^4 {\beta}+16 {\alpha}^4-16 {\alpha}^2 {\beta}^3+432 {\alpha}^2 {\beta}^2-960 {\alpha}^2 {\beta}-320 {\alpha}^2 \\
& \quad \, +48 {\beta}^4-768 {\beta}^3+3456 {\beta}^2-3072 {\beta}+768 \, \, \\
& = 768 - 320 {\alpha}^2 + 16 {\alpha}^4 + 
 4 {\alpha}^6 + (-3072 - 960 {\alpha}^2 - 80 {\alpha}^4) {\beta} \\
& \quad \, + (3456 + 432 {\alpha}^2 + 
    {\alpha}^4) {\beta}^2 + (-768 - 16 {\alpha}^2) {\beta}^3 + 48 {\beta}^4 \\
& = 4 ({\alpha}^2-4)^2 (12 + {\alpha}^2) -16 (192 + 60 {\alpha}^2 + 5 {\alpha}^4) {\beta} + (3456 + 432 {\alpha}^2 + {\alpha}^4) {\beta}^2 \\
& \quad \, -16 (48 + {\alpha}^2) {\beta}^3 + 48 {\beta}^4 \, \, .
\end{align*}
Further, with help of the assumed estimates on $\alpha,\beta$, 
it follows that.
\begin{align*}
\triangle & < 288 - 8192 {\beta} + 6084 {\beta}^2 - 832 {\beta}^3 + 48 {\beta}^4 \\
& = 
4 (72 - 2048 {\beta} + 1521 {\beta}^2 - 208 {\beta}^3 + 12 {\beta}^4) = h({\beta}) \, \, ,
\end{align*}
where $h : {\mathbb{R}} \rightarrow {\mathbb{R}}$ is defined by
\begin{equation*}
h(x) := 4 (72 - 2048 x + 1521 x^2 - 208 x^3 + 12 x^4) \, \, , 
\end{equation*}
for every $x \in {\mathbb{R}}$. We note that, 
\begin{align*}
h^{\prime \prime}(x) & = 12168 - 4992 x + 576 x^2 =  24 (507 - 208 x + 24 x) \\
& = 24 \left[ 
\left(\sqrt{24} \, x - \frac{104}{\sqrt{24}}\right)^2 + \frac{169}{3}
\right] = (24 x - 104)^2 + 1352 > 0 \, \, .
\end{align*}
Hence, $h$ is convex. In addition, 
\begin{equation*}
h\left(\frac{4}{100}\right) < 0 \, \, , \, \, h\left(\frac{165}{100}\right) < 0 
\end{equation*}
and hence 
\begin{equation*}
h(x) < 0 \, \, ,
\end{equation*}
for every 
\begin{equation*}
x \in \left(\frac{4}{100},\frac{165}{100}\right) \, \, .
\end{equation*}
As a consequence, $g$ has $4$ pairwise different roots, 
$2$ of these roots are real, and $2$ of these roots are non-real and conjugate 
complex.
\end{proof} 
In the next step, we find real roots of $q_{\alpha,\beta}$, with the help of the intermediate value
theorem.

\begin{lemma}{(\bf Real roots of $q_{\alpha,\beta}$)} \label{realzeros}
Let $\alpha, \beta \in {\mathbb{R}}$ such that 
\begin{equation*}
4 < \alpha^2 < 6 \, \, , \, \, 0 < \beta < 2
\end{equation*}
and $q_{\alpha,\beta}$ be defined by 
\begin{align*}
q_{\alpha,\beta}(w)  := & \, w^4 + {\alpha} w^3 + ({\beta}-4) w^2 - {\alpha} w + 3  \, \, , \\
= & \,  w^4 + {\alpha} w (w^2 - 1) + ({\beta}-4) w^2 + 3 \, \, ,
\end{align*}
for every $w \in {\mathbb{C}}$.
Then, 
\begin{itemize}
\item[(i)]
if $\alpha > 0$, then $q_{\alpha,\beta}$ has a root in $(-2,-1)$,
\item[(ii)]
if $\alpha < 0$, then $q_{\alpha,\beta}$ has a root in $(1,2)$.
\end{itemize}
\end{lemma}

\begin{proof}
First, we note that
\begin{equation*}
q_{\alpha,\beta}(-1) = q_{\alpha,\beta}(1) = \beta  > 0  
\end{equation*}
and that
\begin{align*}
q_{\alpha,\beta}(t) & = w^4 + {\alpha} w (w^2 - 1) + ({\beta}-4) w^2 + 3 \\
& = w^4 - 4 w^2 + 3 + {\alpha} w (w^2  - 1) + {\beta} w^2 \\
& \leqslant w^4 - 4 w^2 + 3 + {\alpha} w (w^2 - 1) + 2 w^2 \\
& = w^4 - 2 w^2 + 3 + {\alpha} w (w^2 - 1) \, \, ,
\end{align*}
for $w \in {\mathbb{R}}$. 
Hence,  if $\alpha > 0$, then $\alpha > 2$ and
\begin{equation*}
q_{\alpha,\beta}(-2) = 16  - 8 + 3  + \alpha \, (-2)(4 - 1)  = 
11 + 6 (-\alpha) < 11 - 12 = -1 < 0 
\, \, .
\end{equation*}
As a consequence, $q_{\alpha,\beta}$ has a root in $(-2,-1)$.
Further, if $\alpha < 0$, then $\alpha < -2$ and
\begin{equation*}
q_{\alpha,\beta}(2) = 16  - 8 + 3  +  \alpha \cdot 2(4 - 1) =  
11 + 6 \alpha < 11 - 12 = -1 < 0 \, \, .
\end{equation*}
As consequence, $q_{\alpha,\beta}$ has a root in $(1,2)$. 
\end{proof}
Summarizing the obtained information on the roots of $q_{\alpha,\beta}$, 
we obtain:
\begin{theorem}
\label{theorem3}{({\bf Roots of $q_{\alpha,\beta}$})}
Let $\alpha, \beta \in {\mathbb{R}}$ such that 
\begin{equation*}
4 < \alpha^2 < 6 \, \, , \, \, \frac{4}{100} < \beta < \frac{165}{100}  
\, \, .
\end{equation*}
Then $q_{\alpha,\beta}$
\begin{itemize}
\item[(i)]  
has precisely $1$ simple root in 
$U_{1}(0) \cap ({\mathbb{R}} \times (-\infty,0))$,
\item[(ii)] 
$1$ simple root in 
$U_{1}(0) \cap ({\mathbb{R}} \times (0,\infty))$,
\item[(iii)] 
and $2$ different simple roots on ${\mathbb{R}} \setminus 
[-1,1]$.
\end{itemize}
We note that this implies that $q_{\alpha,\beta}$ has no roots on $S^1 \cup [-1,1]$.
\end{theorem}

\begin{proof}
According to Theorem~\ref{theorem1}, {\it the number of roots of $q_{\alpha,\beta}$ in $U_{1}(0)$, multiple roots counted with their multiplicity, is given by
$2$.}  Further,
according Theorem~\ref{theorem2}, 
$q_{\alpha,\beta}$ 
\begin{itemize}
\item[(i)] has $4$ pairwise different roots, 
\item[(ii)] $2$ of these roots are real,
\item[(iii)] and $2$ of these roots are non-real and conjugate 
complex.
\end{itemize} 
As a consequence, $q_{\alpha,\beta}$ has precisely $2$ different roots in 
$U_{1}(0)$.
Also, according to Lemma~\ref{realzeros}, $q_{\alpha,\beta}$
has $1$ real root in ${\mathbb{R}} \setminus [-1,1]$. 
From the assumption that 
$q_{\alpha,\beta}$ has $2$ real roots in $U_{1}(0)$, it follows that 
these roots are different and hence that $q_{\alpha,\beta}$ has $3$ pairwise different
real roots.$\lightning$ Hence there is a non-real root in $U_{1}(0)$. The assumption that there is no root in $U_{1}(0) \cap \left({\mathbb{R}} \times (-\infty,0)\right)$ leads to the existence of $1$ root in  $U_{1}(0) \cap \left({\mathbb{R}} \times (0,\infty)\right)$ 
and hence, since $q_{\alpha,\beta}$ has real coefficients,
to the existence of a root in $U_{1}(0) \cap \left({\mathbb{R}} \times (-\infty,0)\right)$.$\lightning$
Hence, there is a root in $U_{1}(0) \cap \left({\mathbb{R}} \times (-\infty,0)\right)$, and there is also a root in  $U_{1}(0) \cap \left({\mathbb{R}} \times (0,\infty)\right)$. As a consequence, the $2$ real roots are 
contained in ${\mathbb{R}} \setminus (-1,1)$. Since, $q_{\alpha,\beta}(-1) = q_{\alpha,\beta}(1) = \beta > 0$, the 
$2$ real roots are contained in ${\mathbb{R}} \setminus [-1,1]$. We note that this implies that 
there are no roots on $S^1 \cup [-1,1]$. 
\end{proof}

In the final step, we apply Rouch\'e's theorem, to prove the existence
of roots of $p$ in $\Omega_2$, of for $a$ sufficiently close to $M$.

\begin{theorem}{(\bf Roots of $p$)} \label{maintheorem} 
Let 
\begin{equation*}
\alpha = \frac{m}{\mu M} \, \, , \, \, 
\beta = 6 - 
\frac{l(l+1)}{\mu^2 M^2} \, \, ,
\end{equation*}
be such that 
\begin{equation} \label{inequalities}
4 < \alpha^2 < 6 \, \, , \, \, \frac{4}{100} < \beta < \frac{165}{100}  
\, \, .
\end{equation}
Then, for $a$ sufficiently close to $M$, there is 
a root of $p= p_{e} + \delta$ in
\begin{equation*}
U_1(0) \cap ({\mathbb{R}} \times (-\infty,0)) \, \, .
\end{equation*}
\end{theorem}

\begin{proof} 
First, according to Theorem~\ref{theorem3}, $p_{e}$
has precisely $1$ simple root in 
$U_{1}(0) \cap ({\mathbb{R}} \times (-\infty,0))$ and 
no roots in 
\begin{equation*}
C := \partial [U_1(0) \cap ({\mathbb{R}} \times (-\infty,0))] \, \, .
\end{equation*}
Further, 
we note that for every $w \in B_1(0)$:
\begin{align*}
&  |\,\delta(w)| =  
\bigg|- \frac{m \left(1 - \frac{a}{M}\right)}{\mu M} 
\, w (w^{2} - 1) + 2 \, \sqrt{1 - \frac{a^2}{M^2}} \left[1 -i \, \frac{2n+1}{2 \mu M} \, 
\, w (w^{2} - 1)\right] \bigg|
\\
& \leqslant 
|\alpha| \cdot |w| \, (\,|w|^2 + 1) \left(1 - \frac{a}{M}\right) +
2 \, \sqrt{1 - \frac{a^2}{M^2}} 
\left[1 + \frac{2n+1}{2 \mu M} \, |w| \, (|w|^2 + 1)\right]
\\
& \leqslant 2 \sqrt{6} \left(1 - \frac{a}{M}\right) + 
 2 \, \sqrt{1 - \frac{a^2}{M^2}} 
\left(1 + \frac{2n+1}{\mu M}\right) \, \, .
\end{align*}
and hence that 
\begin{equation*} 
\|\,\delta|_{B_1(0)}\|_{\infty} \leqslant  
 2 \sqrt{6} \left(1 - \frac{a}{M}\right) + 
 2 \, \sqrt{1 - \frac{a^2}{M^2}} 
\left(1 + \frac{2n+1}{\mu M}\right) \, \, . 
\end{equation*}
Even further, since there are no roots of $p_{e}$ in $C$,
it follows that 
\begin{equation*}
\frac{1}{|g_{e}|} \bigg|_{C}
\end{equation*}
is continuous function and, since $C$ is compact, that there is $\varepsilon > 0$ such that 
\begin{equation*}
\frac{1}{|g_{e}(w)|} \leqslant \varepsilon \, \, ,
\end{equation*}
for every $w \in C$. The latter implies that 
\begin{equation*}
|g_{e}(w)| \geqslant \frac{1}{\varepsilon} \, \, ,  
\end{equation*}
for every $w \in C$.
Hence for $a/M$ sufficiently close to $1$, it follows that 
\begin{align*}
 2 \sqrt{6} \left(1 - \frac{a}{M}\right) + 
 2 \, \sqrt{1 - \frac{a^2}{M^2}} 
\left(1 + \frac{2n+1}{\mu M}\right)  < \frac{1}{\varepsilon} 
\end{align*}
and hence that 
\begin{equation*}
|\delta(w)| < |g_{e}(w)| \, \, ,
\end{equation*}
for every $w \in C$. Hence for such a case, it follows from Rouch\'e's theorem
that there is 
a root of $p = p_{e} + \delta$ in $C$.
\end{proof}

The following proposition rewrites the inequalities (\ref{inequalities}) in terms
of the parameters $\mu, M, m$ and $l$.
\begin{prop}
If 
\begin{equation*}
\frac{25}{149} \, l(l+1) < \mu^2 M^2 < \frac{20}{87} \, m^2 \, \, , 
\end{equation*}
then 
\begin{equation*}
4 < \alpha^2 < 6 \, \, , \, \, \frac{4}{100} < \beta < \frac{165}{100}  
\, \, ,
\end{equation*}
where 
\begin{equation*}
\alpha = \frac{m}{\mu M} \, \, , \, \, 
\beta = 6 - 
\frac{l(l+1)}{\mu^2 M^2} \, \, .
\end{equation*}
If $l = |m| + k$, where $k \in {\mathbb{N}}$, the interval 
\begin{equation*}
\left(\, \frac{25}{149} \, l(l+1) \, , \, \frac{20}{87} \, m^2 \right) 
\end{equation*}
is non-empty, iff
\begin{equation} \label{condition}
|m| > \frac{435}{161} \, (k + 1) + \sqrt{\left(\frac{435}{161}\right)^2 (k+1)^2
+ \frac{435}{161} \, k(k+1)} \, \, .
\end{equation}
We note that 
if $k=0$, (\ref{condition}) leads to 
\begin{equation*}
|m| > 2 \, \frac{435}{161} = \frac{870}{161} \approx 5.40373 \, \, .
\end{equation*}
\end{prop}

\begin{proof}
Since 
\begin{equation*}
\alpha = \frac{m}{\mu M} \, \, , \, \, 
\beta = 6 - 
\frac{l(l+1)}{\mu^2 M^2} \, \, ,
\end{equation*}
the inequalities
\begin{equation*}
4 < \alpha^2 < 6 \, \, \wedge \, \, \frac{4}{100} < \beta < \frac{165}{100}  
\end{equation*}
are equivalent to 
\begin{equation} \label{inequalities}
4 < \frac{m^2}{\mu^2 M^2} < 6 \, \, \wedge \, \, \frac{4}{100} < 6 - 
\frac{l(l+1)}{\mu^2 M^2} < \frac{165}{100}  
\, \, .
\end{equation}
We note the equivalence of the following inequalities
\begin{align*}
&  \frac{4}{100} < 6 - \frac{l(l+1)}{\mu^2 M^2} < \frac{165}{100} \, \, \Leftrightarrow \, \, 
- \frac{4}{100} >  - 6 +  \frac{l(l+1)}{\mu^2 M^2}  > - \frac{165}{100} \, \, , \\
&  6 - \frac{4}{100} >  \frac{l(l+1)}{\mu^2 M^2} > 6 - \frac{165}{100} \, \,\Leftrightarrow \, \,
\frac{149}{25} >   \frac{l(l+1)}{\mu^2 M^2}  > \frac{87}{20} \, \, , \\
& \frac{25}{149} <   \frac{\mu^2 M^2}{l(l+1)}  < \frac{20}{87} \, \,
\Leftrightarrow \, \, 
\frac{25}{149} \, l(l+1) <   \mu^2 M^2  < \frac{20}{87} \, l(l+1) 
\end{align*}
as well as
\begin{align*}
& 4 < \frac{m^2}{\mu^2 M^2} < 6 \, \, \Leftrightarrow \, \, 
\frac{1}{6} < \frac{\mu^2 M^2}{m^2} < \frac{1}{4} \, \, \Leftrightarrow
\frac{m^2}{6} < \mu^2 M^2 < \frac{m^2}{4}  \, \, .
\end{align*}
Hence, (\ref{inequalities}) is equivalent to 
\begin{equation*}
\frac{25}{149} \, l(l+1) < \mu^2 M^2 <
\min\left\{\frac{m^2}{4},\frac{20}{87} \, l(l+1)\right\}  =
\min\left\{\frac{m^2}{4},\frac{l(l+1)}{4.35}\right\}
\, \, , 
\end{equation*}
where we used that, since $l \geqslant |m|$, 
\begin{equation*}
\frac{25}{149} \, l(l+1) \geqslant \frac{25}{149} \, |m|^2 \geqslant
\frac{25}{150} \, |m|^2 = \frac{m^2}{6} \, \, .
\end{equation*}
Since
\begin{equation*}
\frac{m^2}{4.35} =
\min\left\{\frac{m^2}{4},\frac{m^2}{4.35}\right\} \leqslant
\min\left\{\frac{m^2}{4},\frac{l(l+1)}{4.35}\right\}  \, \, ,
\end{equation*}
it follows that the inequality  
\begin{equation*}
\frac{25}{149} \, l(l+1) < \mu^2 M^2 < \frac{20}{87} \, m^2 
\end{equation*}
implies (\ref{inequalities}).
If $l = |m| + k$, where $k \in {\mathbb{N}}$, the interval 
\begin{equation*}
\left(\, \frac{25}{149} \, l(l+1) \, , \, \frac{20}{87} \, m^2 \right) 
\end{equation*}
is non-empty, iff
\begin{align*}
& 
\frac{20}{87} \, m^2 > \frac{25}{149} \, (|m| + k)(|m| + k +1)
\, \, , \\
& 
\frac{20}{87} \, m^2 > \frac{25}{149} \, (|m| + k)(|m| + k +1) =
\frac{25}{149} \, [m^2 + 2 |m| (k + 1) + k(k+1)] 
\, \, , \\
& \frac{805}{12963} \, m^2  > \frac{25}{149} \, [2 |m| (k + 1) + k(k+1)] \, \, , \\
& m^2 > \frac{435}{161} \, [2 |m| (k + 1) + k(k+1)] \, \, , \\
& m^2 - \frac{435}{161} \, [2 |m| (k + 1) + k(k+1)] > 0 \, \, , \\
& \left[|m| - \frac{435}{161} \, (k + 1)\right]^2 - 
\left(\frac{435}{161}\right)^2 (k+1)^2
- \frac{435}{161} k(k+1) > 0 \, \, , \\
& \left[|m| - \frac{435}{161} \, (k + 1)\right]^2 > 
\left(\frac{435}{161}\right)^2 (k+1)^2
+ \frac{435}{161} \, k(k+1) \, \, .
\end{align*}
and hence iff 
\begin{equation*}
|m| > \frac{435}{161} \, (k + 1) + \sqrt{\left(\frac{435}{161}\right)^2 (k+1)^2
+ \frac{435}{161} \, k(k+1)} \, \, .
\end{equation*}
\end{proof}

\section{Approach~2}

Approach~2 uses the subsequent conformal transformation $h$ to transform the open 
lower half-disk $\Omega_2$ onto to the first quadrant $(0,\infty)^2$. The roots of 
$p \circ h^{-1}$ coincide with the roots of the fourth order polynomial $q$, given in 
Definition~\ref{definitionofq}. Subsequently, the argument principle is used to derive 
Theorem~\ref{instabilitytheorem}. Lemmatas~\ref{basiclemma1} and \ref{basiclemma2} prepare
the proof of Theorem~\ref{instabilitytheorem}. Theorem~\ref{maintheorem} shows the existence 
of roots of $p$ in $\Omega_2$ for $a/M$ satisfying the inequality (\ref{newinstabilitycondition2}), i.e., for values down to about $0.979796$.

\begin{lemma}[{\bf A biholomorphic map from the open lower half-disk onto the open 
first quadrant}]
\label{biholomorphicmap}
By 
\begin{equation*}
h(z) := \frac{1 - z}{1 + z}  \, \, ,
\end{equation*}
for every $z \in U_{1}(0) 
\cap ({\mathbb{R}}  \times (-\infty,0))$, there is 
defined a biholomorphic map 
\begin{equation*}
h : U_{1}(0) 
\cap ({\mathbb{R}}  \times (-\infty,0)) \rightarrow 
(0,\infty)^2 \, \, ,
\end{equation*}
with inverse 
\begin{equation*}
h^{-1} :
(0,\infty)^2 \rightarrow U_{1}(0) 
\cap ({\mathbb{R}}  \times (-\infty,0)) \, \, ,
\end{equation*}
defined by 
\begin{equation*}
h^{-1}(u) = \frac{1 - u}{1 + u} \, \, , 
\end{equation*}
for every $u \in (0,\infty)^2$. 
\end{lemma}

\begin{proof}
If $z \in U_{1}(0) 
\cap ({\mathbb{R}}  \times (-\infty,0))$, $x := \textrm{Re}(z)$ and $y := 
\textrm{Im}(z) \, (< 0)$, then
\begin{align*}
& \frac{1 - z}{1 + z} = \frac{1 - x - i y}{1 + x + iy} = 
\frac{(1 - x - i y)(1 + x - iy)}{(1 + x + iy)(1 + x - iy)} \\
& = \frac{1 - x^2 - y^2 - 2 i y}{(1 + x)^2 + y^2} \in (0,\infty)^2 \, \, .
\end{align*}
Hence by
\begin{equation*}
h(z) := \frac{1 - z}{1 + z}  \, \, ,
\end{equation*}
for every $z \in U_{1}(0) 
\cap ({\mathbb{R}}  \times (-\infty,0))$, there is 
defined a holomorphic map 
\begin{equation*}
h : U_{1}(0) 
\cap ({\mathbb{R}}  \times (-\infty,0)) \rightarrow 
(0,\infty)^2 \, \, .
\end{equation*}
Further, if $u \in (0,\infty)^2$, $u_1 := \textrm{Re}(u) \, (> 0)$ and
$u_2 := \textrm{Im}(u) \, (> 0)$, then 
\begin{equation*}
\frac{1 - u}{1 + u} = \frac{1 - u_1^2 - u_2^2 - 2 i u_2}{(1 + u_1)^2 + u_2^2}
\in {\mathbb{R}}  \times (-\infty,0) \, \, .
\end{equation*}
In addition, 
\begin{align*}
& \left[\frac{1 - u_1^2 - u_2^2}{(1 + u_1)^2 + u_2^2}\right]^2 + 
\left[\frac{- 2 u_2}{(1 + u_1)^2 + u_2^2} \right]^2 \\
& =
\frac{(1 - u_1^2 - u_2^2)^2 + 4 u_2^2}{[(1 + u_1)^2 + u_2^2]^2} =
\frac{(1 - u_1^2 - u_2^2)^2 + 4 u_2^2}{(1 + u_1^2 + u_2^2 + 2 u_1)^2} \\
& = \frac{(1 + u_1^2 + u_2^2)^2 - 4 (u_1^2 + u_2^2) + 4 u_2^2}{(1 + u_1^2 + u_2^2)^2 + 4 u_1(1 + u_1^2 + u_2^2) + 4 u_1^2} \\
& = \frac{(1 + u_1^2 + u_2^2)^2 - 4 u_1^2}{(1 + u_1^2 + u_2^2)^2 + 4 u_1(1 + u_1^2 + u_2^2) + 4 u_1^2} \\
& <
\frac{(1 + u_1^2 + u_2^2)^2 - 4 u_1^2}{(1 + u_1^2 + u_2^2)^2}
< \frac{(1 + u_1^2 + u_2^2)^2}{(1 + u_1^2 + u_2^2)^2} = 1 \, \, , 
\end{align*}
and hence 
\begin{equation*}
\frac{1 - u}{1 + u} \in U_1(0) \, \, .
\end{equation*}
As a consequence, by 
\begin{equation*}
g(u) := \frac{1 - u}{1 + u} \, \, , 
\end{equation*}
for every $u \in (0,\infty)^2$, there is defined a holomorphic map
\begin{equation*} 
g : (0,\infty)^2 \rightarrow U_{1}(0) 
\cap ({\mathbb{R}}  \times (-\infty,0)) \, \, .
\end{equation*}
Further, for every $z \in U_{1}(0) 
\cap ({\mathbb{R}}  \times (-\infty,0))$,
\begin{equation*}
g(h(z)) = g\left(\frac{1 - z}{1 + z}\right) = 
\frac{1 - \frac{1 - z}{1 + z}}{1 + \frac{1 - z}{1 + z}} =
\frac{1 + z - 1 + z}{1 + z + 1 - z} = \frac{2z}{2} = z 
\end{equation*}
as well as 
\begin{equation*}
h(g(u)) = h\left(\frac{1 - u}{1 + u}\right) = 
\frac{1 - \frac{1 - u}{1 + u}}{1 + \frac{1 - u}{1 + u}} =
\frac{1 + u - 1 + u}{1 + u + 1 - u} = \frac{2u}{2} = u \, \, , 
\end{equation*}
for every  $u \in (0,\infty)^2$.
\end{proof}

We note that 
\begin{align*}
p(w) & = 
w^4 + \frac{\frac{ma}{M} - i (2n + 1)\sqrt{1 - \frac{a^2}{M^2}}}{\mu M} 
\, w (w^{2} - 1) \nonumber \\
& \quad \, \, \, + \left[ 2 - 
\frac{l(l+1)}{\mu^2 M^2} \right] w^{2}  
+ 3 + 2 \sqrt{1 - \frac{a^2}{M^2}} \\
& = w^4 + \alpha w \left(w^2 - 1\right) + (\beta -4) w^2 + 3 + \epsilon \, \, ,
\end{align*}
for every $w \in \Omega_2 =  U_{1}(0) \cap ({\mathbb{R}}  \times (-\infty,0))$, 
where 
\begin{align*}
& \alpha := \frac{\frac{ma}{M} - i (2n + 1)\sqrt{1 - \frac{a^2}{M^2}}}{\mu M}  \, \, , \, \, 
\beta := 6 - 
\frac{l(l+1)}{\mu^2 M^2} \, \, , \, \,  \epsilon := 2 \sqrt{1 - \frac{a^2}{M^2}} 
\, \, (> 0) \, \, .
\end{align*}
In particular, 
\begin{equation*}
\alpha = \alpha_1 - i \alpha_2 \, \, , 
\end{equation*}
where 
\begin{equation*}
\alpha_1 = \frac{ma}{\mu M^2}  \, \, , \, \, 
\alpha_2 = \frac{2n + 1}{\mu M} \, \sqrt{1 - \frac{a^2}{M^2}} \, \, (> 0) \, \, .
\end{equation*}
Further, with the help of the biholomorphic map $h$ from 
Proposition~\ref{biholomorphicmap}, it follows that 
\begin{align*}
& (p \circ h^{-1})(z) \\
& = \frac{1}{(1+z)^4}\left[(\beta +\epsilon) z^4  + 4 (2+\alpha +\epsilon) z^3 + 
2 (16 - \beta + 3 \epsilon) z^2 + 
4  (2 - \alpha + \epsilon) 
z + 1 \right]
\end{align*}
for every $z \in (0,\infty)^2$.
Hence,  we make the following 
\begin{definition}{(\bf Definition of $q$)}
We define for $\beta \geqslant 0$ 
\begin{equation}  \label{definitionofq}
q(z) := z^4  + \frac{4 z}{\beta +\epsilon} \left[
(2 + \epsilon)(z^2+1) + \alpha (z^2-1)\right] 
+ 
2 \, \frac{16 - \beta + 3 \epsilon}{\beta +\epsilon} \, z^2 + 1 
\, \, , 
\end{equation}
for every $z \in {\mathbb{C}}$.
\end{definition}

\begin{figure} 
\centering
\includegraphics[width=5.6cm,height=5.2752cm]{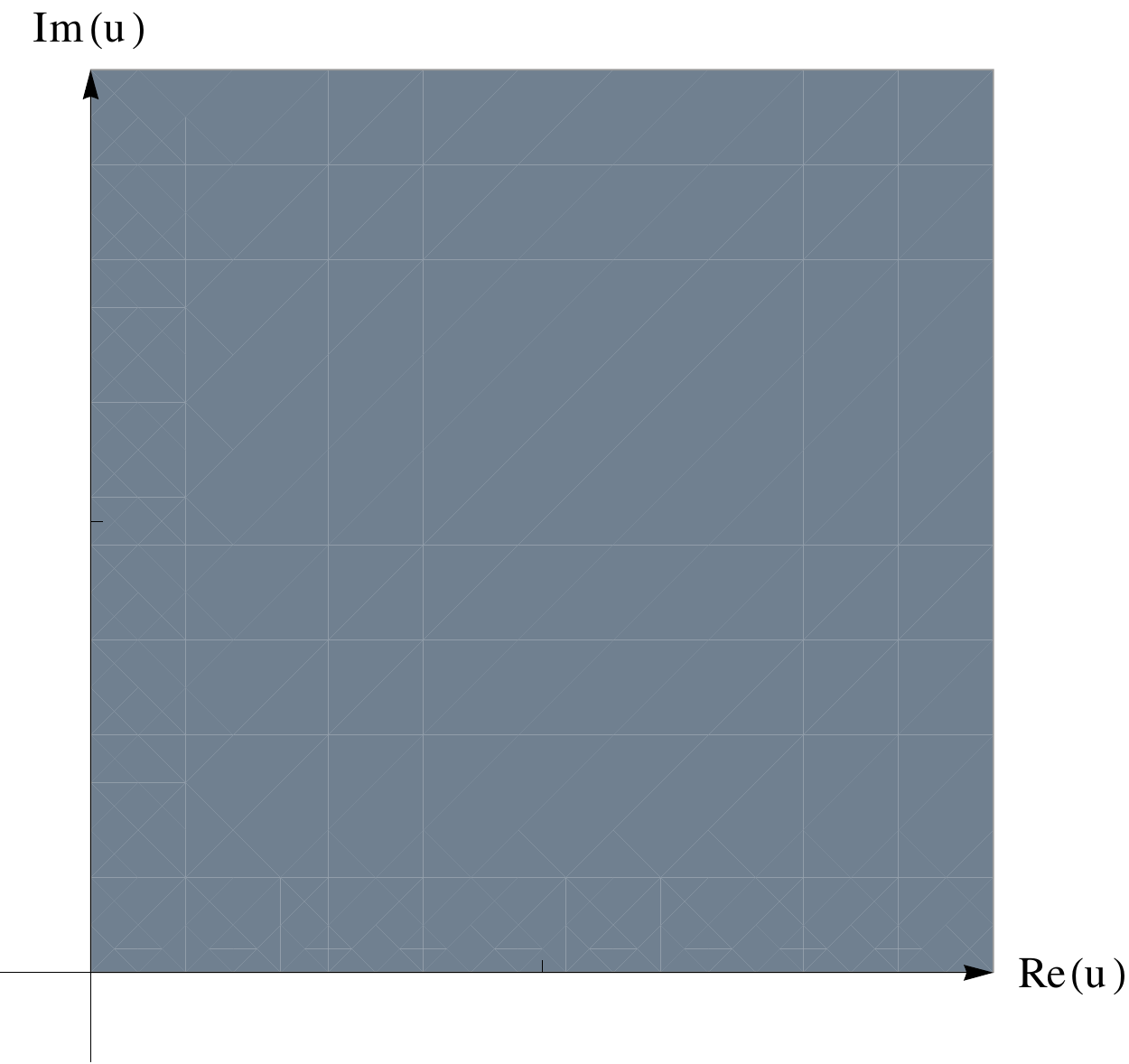}
\caption{The domain of values of $u$ inside the domain of 
$q$, leading on unstable $\lambda$ is given by the open first quadrant, shaded in gray.}.
\label{fig4}
\end{figure}

\begin{lemma}{(\bf Instability in terms of roots of $q$)}
If $R = r_{+}$, i.e., $R_{-} = 2 (M^2 - a^2)^{1/2}$, 
and $\lambda$ satisfies (\ref{exceptions1}),
then $\lambda \in {\mathbb{R}} \times (-\infty,0)$ is such 
that   
$\ker(A - \lambda B - \lambda^2)$ is non-trivial, if and only if  
\begin{equation*}
\lambda = - \mu \, \frac{1 + z^2}{1 - z^2} \, \, , 
\end{equation*}
for a root $z$ of $q$ contained in the open first quadrant, $(0,\infty)^2$. 
\end{lemma}

\begin{lemma} \label{basiclemma1}
The polynomial $q$ has no real roots. In addition, if 
\begin{equation*}
\alpha_{1} \geqslant 0 \, \, \wedge  \, \, 0 \leqslant \beta \leqslant 16 + 3 \epsilon \, \, ,  
\end{equation*}
then 
\begin{equation*}
\lim_{x \rightarrow \infty} 
\arctan\left(\frac{\textrm{Im}(q(x))}{\textrm{Re}(q(x))}\right) = 0 \, \, .
\end{equation*} 
\end{lemma}

\begin{proof}
It follows that
\begin{align*}
q(x) & = x^4  + \frac{4 x}{\beta +\epsilon} \left[
(\epsilon + 2)(x^2+1) + \alpha (x^2-1)\right] 
+ 
2 \, \frac{16 - \beta + 3 \epsilon}{\beta +\epsilon} \, x^2 + 1 \\
& =  x^4  + \frac{4 x}{\beta +\epsilon} \left[
(\epsilon + 2)(x^2+1) + \alpha_1 (x^2-1)\right] 
+ 
2 \, \frac{16 - \beta + 3 \epsilon}{\beta +\epsilon} \, x^2 + 1  \\
& \quad \, + i \frac{4 \alpha _2}{\beta +\epsilon} \, x (1 - x^2) \\
& 
=  x^4  + \frac{4 x}{\beta +\epsilon} \left\{
(\alpha_1 + \epsilon + 2) x^2 - [\alpha_1 -(\epsilon + 2)] \right\} 
+ 
2 \, \frac{16 - \beta + 3 \epsilon}{\beta +\epsilon} \, x^2 + 1  \\
& \quad \, + i \frac{4 \alpha _2}{\beta +\epsilon} \, x (1 - x^2)
\, \, ,
\end{align*}
for very $x \in {\mathbb{R}}$. We note that $q$ has no real roots. This can be seen as follows. If $x$ is a real root of $q$, then $\textrm{Im}(q(x)) = 0$
and hence $x \in \{-1,0,1\}$. Further, 
\begin{align*}
& \textrm{Re}(q(0)) = 1 \neq 0 \, \, , \\
& \textrm{Re}(q(-1)) = \textrm{Re}(q(1)) \\
& =
1  + \frac{4}{\beta +\epsilon} \left\{
(\alpha_1 + \epsilon + 2) - [\alpha_1 -(\epsilon + 2)] \right\} 
+ 
2 \, \frac{16 - \beta + 3 \epsilon}{\beta +\epsilon} + 1 \\
& = 2 + \frac{8(\epsilon + 2)}{\beta +\epsilon} + 2 \, \frac{16 - \beta + 3 \epsilon}{\beta +\epsilon} = \frac{2(\beta +\epsilon)}{\beta +\epsilon} + \frac{8(\epsilon + 2)}{\beta +\epsilon} + 2 \, \frac{16 - \beta + 3 \epsilon}{\beta +\epsilon} \\
& = \frac{2(\beta +\epsilon) + 8(\epsilon + 2)+ 2(16 - \beta + 3 \epsilon)}{\beta +\epsilon} = 16 \, \frac{3 + \epsilon}{\beta +\epsilon} \neq 0 \, \, .
\end{align*}
If 
\begin{equation*}
\alpha_{1} \geqslant 0 \, \, \wedge  \, \, 0 \leqslant \beta \leqslant 16 + 3 \epsilon \, \, ,  
\end{equation*}
then 
\begin{equation*}
\textrm{Re}(q(x)) \geqslant x^4 + 1 \, \, , \, \, \text{if $x > 1$}
\end{equation*}
as well as 
\begin{equation*}
\textrm{Im}(q(x)) = \frac{4 \alpha _2}{\beta +\epsilon} \, x (1 - x^2) <
0 \, \, ,  \, \, \text{if $x > 1$} \, \, ,
\end{equation*}
implying that 
\begin{equation*}
0 > \frac{\textrm{Im}(q(x))}{\textrm{Re}(q(x))} 
\geqslant   \frac{4 \alpha _2}{\beta +\epsilon} \, \frac{x (1 - x^2)}{ x^4 + 1}
\, \, \text{if $x > 1$} 
\end{equation*}
and hence that 
\begin{equation*}
\lim_{x \rightarrow \infty}  \frac{\textrm{Im}(q(x))}{\textrm{Re}(q(x))}  = 0 \, \, , \, \,  \lim_{x \rightarrow \infty} 
\arctan\left(\frac{\textrm{Im}(q(x))}{\textrm{Re}(q(x))}\right) = 0 \, \, .
\end{equation*}
\end{proof}

\begin{lemma} \label{basiclemma2}
For 
\begin{equation*}
\alpha_1 > 2 + \epsilon \lor \alpha_1 <  - (2 + \epsilon) \, \, , 
\end{equation*}
the polynomial 
$q$ has no purely imaginary roots.
\begin{itemize}
\item[(i)] If
\begin{equation*}
\alpha_1 > 2 + \epsilon \, \, , 
\end{equation*}
the function 
\begin{equation*}
h_2 := ({\mathbb{R}} \rightarrow {\mathbb{R}} ,
y \mapsto {\textrm{Im}}(q(iy)) )
\end{equation*}
is strictly decreasing and 
\begin{equation*}
h_2((0,\infty)) \subset (-\infty,0) \, \, . 
\end{equation*}
\item[(ii)]
If
\begin{equation*}
\alpha_1 < -(2 + \epsilon) \, \, , 
\end{equation*}
the function $h_2$
is strictly increasing and 
\begin{equation*}
h_2((0,\infty)) \subset (0,\infty) \, \, . 
\end{equation*}
\item[(iii)] If
\begin{equation*}
0 < \beta < 8 + \epsilon \, \, ,
\end{equation*}
then 
\begin{equation*}
h_1 := ({\mathbb{R}} \rightarrow {\mathbb{R}} ,
y \mapsto {\textrm{Re}}(q(iy)))
\end{equation*}
has precisely $2$ positive roots $y_0,y_1$, satisfying 
$0 < y_0 < y_1$. In addition, 
\begin{equation*}
\begin{cases}
h_1(y) > 0 & \text{for $0 \leqslant y < {\bar{y}}_{0}$} \\
h_1({\bar{y}}_{0}) = 0 \text{} \\
h_1(y) < 0 & \text{for ${\bar{y}}_{0} < y < {\bar{y}}_{1}$} \\
h_1({\bar{y}}_{1}) = 0 \text{} \\
h_1(y) > 0 & \text{for $y > {\bar{y}}_{1}$} 
\end{cases}
\, \, .
\end{equation*}
\end{itemize}  
\end{lemma}

\begin{proof}
It follows that
\begin{align*}
q(iy) 
& = 
y^4 - 2 \, \frac{16 - \beta + 3 \epsilon}{\beta +\epsilon } \, y^2 + 1 - \frac{4 \alpha _2}{\beta +\epsilon } \, y (y^2 + 1) \\
& \quad \, 
-i \frac{4 y}{\beta +\epsilon} \left[ (\epsilon + 2) \left(y^2 -1\right) +
\alpha _1 \left(y^2 + 1\right) \right]  \\
& = \left(y^2-\frac{16-\beta +3 \epsilon }{\beta +\epsilon }\right)^{2} - \frac{8 (8 + \epsilon  - \beta ) (4+\epsilon )}{(\beta +\epsilon )^2}
 - \frac{4 \alpha _2}{\beta +\epsilon } \, y (y^2 + 1 ) \\
& \quad \, 
\, -i \frac{4 y}{\beta +\epsilon} \left[ (\alpha_1 + \epsilon + 2) y^2 +
(\alpha _1 - (\epsilon + 2)) \right]  \, \, ,  
\end{align*}
for every $y \in {\mathbb{R}}$.
We note that for 
\begin{equation*}
\alpha_1 > 2 + \epsilon \lor \alpha_1 <  - (2 + \epsilon) \, \, , 
\end{equation*}
$q$ has no purely imaginary roots. This can be seen as follows. 
If $i y$, where $y \in {\mathbb{R}}$ is a purely imaginary root of $q$, then 
$\textrm{Im}(q(iy)) = 0$ and hence $y = 0$.
On the other hand, 
\begin{equation*}
\textrm{Re}(q(0)) = 1 \neq 0 \, \, . 
\end{equation*}
We note, for 
\begin{equation*}
\alpha_1 > 2 + \epsilon \, \, , 
\end{equation*}
that the function 
\begin{equation*}
h_2 := ({\mathbb{R}} \rightarrow {\mathbb{R}} ,
y \mapsto {\textrm{Im}}(q(iy)) )
\end{equation*}
is strictly decreasing, since 
\begin{equation*}
h_2^{\prime}(y) = -\frac{4}{\beta +\epsilon }\left\{\left[\alpha _{1 }-(2+ \epsilon )\right]+3 \left[\alpha _{1 }+(2+ \epsilon )\right] y^2 \right\} < 0 \, \, , 
\end{equation*}
for every $y \in {\mathbb{R}}$.
Also, it follows for $y > 0$, that 
\begin{align*}
h_2(y) = \textrm{Im}(q(iy)) = - \frac{4 y}{\beta +\epsilon} \left[ (\alpha_1 + \epsilon + 2) y^2 +
(\alpha _1 - (\epsilon + 2)) \right]  < 0 \, \, .
\end{align*}
Analogously, 
for 
\begin{equation*}
\alpha_1 <  -(2 + \epsilon) \, \, , 
\end{equation*}
the function $h_2$
is strictly increasing, since 
\begin{equation*}
h_2^{\prime}(y) = -\frac{4}{\beta +\epsilon }\left\{\left[\alpha _{1 }-(2+ \epsilon )\right]+3 \left[\alpha _{1 }+(2+ \epsilon )\right] y^2 \right\} > 0 \, \, , 
\end{equation*}
for every $y \in {\mathbb{R}}$.
Also, it follows for $y > 0$, that 
\begin{align*}
h_2(y) = \textrm{Im}(q(iy)) = - \frac{4 y}{\beta +\epsilon} \left[ (\alpha_1 + \epsilon + 2) y^2 +
(\alpha _1 - (\epsilon + 2)) \right] > 0 \, \, .
\end{align*}
We note, for
\begin{equation*}
0 < \beta < 8 + \epsilon \, \, ,
\end{equation*}
that the function 
\begin{equation*}
h_1 := ({\mathbb{R}} \rightarrow {\mathbb{R}} ,
y \mapsto {\textrm{Re}}(q(iy)) )
\end{equation*}
is strictly decreasing on 
\begin{equation*}
I_1 := \left(0\,,\sqrt{\frac{16-\beta +3 \epsilon}{\beta +\epsilon }}\,\right) \, \, ,
\end{equation*} 
since 
\begin{align*}
h_1^{\prime}(y) =
4y\left(y^2-\frac{16-\beta +3 \epsilon}{\beta +\epsilon }\right)-\frac{4 \alpha _2\left(3 y^2+1\right) }{\beta +\epsilon} < 0 \, \, , 
\end{align*}
for every $y \in I_1$, where 
\begin{align*}
& \sqrt{\frac{16-\beta +3 \epsilon}{\beta +\epsilon }} =
\sqrt{\frac{16 + 4 \epsilon - (\beta + \epsilon)}{\beta +\epsilon }} =
\sqrt{\frac{16 + 4 \epsilon}{\beta +\epsilon} - 1} > 
\sqrt{\frac{16 + 4 \epsilon}{8 + 2 \epsilon} - 1}  = 1 \, \, .
\end{align*}
Further, we note that the following inequalities are equivalent:
\begin{align*}
& 16-\beta +3 \epsilon > 2 \sqrt{2} \sqrt{(4+\epsilon ) (8-\beta +\epsilon )} \, \, , \\
& (16-\beta + 3 \epsilon)^2 > 8 (4+\epsilon ) (8-\beta +\epsilon ) \, \, , \\
& [8 -\beta + \epsilon + 2(4 + \epsilon)]^2 > 8 (4+\epsilon ) (8-\beta +\epsilon ) \, \, , \\
& (8 -\beta + \epsilon)^2 + 4 (4+\epsilon ) (8-\beta +\epsilon) + 4(4 + \epsilon)^2
> 8 (4+\epsilon ) (8-\beta +\epsilon ) \\
& (8 -\beta + \epsilon)^2 - 4 (4+\epsilon ) (8-\beta +\epsilon) + 4(4 + \epsilon)^2
> 0 \, \, , \\
& [8 -\beta + \epsilon - 2(4 + \epsilon)]^2 > 0 \, \, .
\end{align*}
Since
\begin{equation*}
y_j^4 - 2 \, \frac{16 - \beta + 3 \epsilon}{\beta +\epsilon } \, y_j^2 + 1
= 0 \, \, , 
\end{equation*}
for $j \in \{0,1\}$,
where
\begin{align*}
(\, 0 < \, ) \, \, y_0 & := \sqrt{\frac{16-\beta +3 \epsilon -2 \sqrt{2} \sqrt{(4+\epsilon ) (8-\beta +\epsilon )}}{\beta +\epsilon }} < \sqrt{\frac{16-\beta +3 \epsilon }{\beta +\epsilon }} \, \, , \\
y_1 & := \sqrt{\frac{16-\beta +3 \epsilon + 2 \sqrt{2} \sqrt{(4+\epsilon ) (8-\beta +\epsilon )}}{\beta +\epsilon }} > \sqrt{\frac{16-\beta +3 \epsilon }{\beta +\epsilon }} \, \, , 
\end{align*}
it follows that 
\begin{align*}
h_1(y_0) = - \frac{4 \alpha _2}{\beta +\epsilon } \, y_0 (y_0^2 + 1 ) < 0 \, \, , 
\, \, h_1(y_1) = - \frac{4 \alpha _2}{\beta +\epsilon } \, y_1 (y_1^2 + 1 ) < 0 \, \, .
\end{align*}
Since $h_1(0) = 1 > 0$, there is ${\bar{y}}_{0} \in (0,y_0) \subset I_1$ such that 
\begin{equation*}
h_1({\bar{y}}_{0}) = 0 \, \, .
\end{equation*}
As a consequence of the fact that
$h_{1}$ is strictly decreasing on $I_1$, it follows that 
\begin{equation*}
\begin{cases}
h_1(y) > 0 & \text{for $0 \leqslant y < {\bar{y}}_{0}$} \\
h_1({\bar{y}}_{0}) = 0 \text{} \\
h_1(y) < 0 & \text{for $y \in I_1$ such that $y > {\bar{y}}_{0}$} 
\end{cases}
\, \, .
\end{equation*}
Since, for $y \in {\mathbb{R}}$, 
\begin{align*}
& h_1(y) = 
y^4 - 2 \, \frac{16 - \beta + 3 \epsilon}{\beta +\epsilon } \, y^2 + 1 - \frac{4 \alpha _2}{\beta +\epsilon } \, y (y^2 + 1) \, \, ,
\end{align*}
for sufficiently large $y > 0$, is dominated by the highest power, i.e., $4$, 
there is $\xi_1 > y_1$, such that $h_1(\xi_1) > 0$. Hence there is 
${\bar{y}}_1 \in (y_1,\xi_1)$, such that 
\begin{equation*}
h_1({\bar{y}}_1) = 0 \, \, .
\end{equation*}
We note that the discriminant $\triangle$ of $h_1$ is given by
\begin{align*}
\triangle = 
\frac{4096}{(\beta +\epsilon )^6} \left(8+\epsilon -\beta +2 \alpha _2\right) \left(8+\epsilon -\beta -2 \alpha _2\right) \left(8 \beta +8 \epsilon +2 \beta  \epsilon +2 \epsilon ^2+\alpha _2^2\right)^2 \, \, .
\end{align*}
Hence, if 
\begin{equation*}
\alpha _2 < \frac{1}{2} \, (8 + \epsilon -\beta) \, \, , 
\end{equation*}
then 
\begin{equation*}
\triangle > 0 
\end{equation*}
and $h_1$ has only real roots. In these cases Descartes' rule of signs 
is exact, see, e.g., Corollary~10.1.12 in \cite{rahman}. Since, 
\begin{equation*}
h_1(-y) = 
y^4 - 2 \, \frac{16 - \beta + 3 \epsilon}{\beta +\epsilon } \, y^2 + 1 + \frac{4 \alpha _2}{\beta +\epsilon } \, y (y^2 + 1) \, \, ,
\end{equation*}
for every $y \in {\mathbb{R}}$, and there are $2$ sign changes 
in the previous polynomial, this polynomial has precisely $2$ positive roots. 
As a consequence, $h_1$ has precisely $2$ negative roots and $2$ positive roots, 
the latter given by ${\bar{y}}_{0}$ and ${\bar{y}}_{1}$. 
If 
\begin{equation*}
\alpha _2 = \frac{1}{2} \, (8 + \epsilon -\beta) \, \, ,
\end{equation*}
then 
\begin{equation*}
\triangle = 0 \, \, , 
\end{equation*}
and $h_1$ has in addition the double root $-1$.
If 
\begin{equation*}
\alpha _2 < \frac{1}{2} \, (8 + \epsilon -\beta) \, \, ,
\end{equation*}
then 
\begin{equation*}
\triangle < 0 \, \, , 
\end{equation*}
and $h_1$ has in addition $2$ conjugate complex roots. 
Hence, in all these cases, $h_1$ has precisely $2$ positive roots, 
the latter given by ${\bar{y}}_{0}$ and ${\bar{y}}_{1}$. As a consequence, 
\begin{equation*}
\begin{cases}
h_1(y) > 0 & \text{for $0 \leqslant y < {\bar{y}}_{0}$} \\
h_1({\bar{y}}_{0}) = 0 \text{} \\
h_1(y) < 0 & \text{for ${\bar{y}}_{0} < y < {\bar{y}}_{1}$} \\
h_1({\bar{y}}_{1}) = 0 \text{} \\
h_1(y) > 0 & \text{for $y > {\bar{y}}_{1}$} 
\end{cases}
\, \, .
\end{equation*}
\end{proof}

\begin{theorem} \label{instabilitytheorem}
If 
\begin{equation} \label{instabilitycondition}
\alpha_1 > 2 + \epsilon \, \, \wedge \, \, 
0 < \beta < 8 + \epsilon \, \, ,
\end{equation}
then the open first quadrant contains precisely $1$ root of $q$.
\end{theorem}

\begin{proof}
For the proof, we use the argument principle. We consider $q$ on the 
intersection $D$ of $U_{R}(0)$ with the open first quadrant, where 
$R > 0$ is sufficiently large. As a consequence of the conditions 
(\ref{instabilitycondition}) and according to Lemmas~\ref{basiclemma1},~\ref{basiclemma2}, 
there are no roots of $q$ on the boundary of $D$.
For $R > 0$ and $\theta \in [0,\pi/2]$, it follows that 
\begin{align*}
q\left(R e^{i \theta}\right) & = \left(R e^{i \theta}\right)^4  + \frac{4 R e^{i \theta}}{\beta +\epsilon} \left\{
(\alpha_1 + \epsilon + 2) \left(R e^{i \theta}\right)^2 - [\alpha_1 -(\epsilon + 2)] \right\} \\
& \, \, \, \, \, \, \, \,  
+ 
2 \, \frac{16 - \beta + 3 \epsilon}{\beta +\epsilon} \, \left(R e^{i \theta}\right)^2 + 1 \\
& = R^4  e^{4 i \theta} + \frac{4 R e^{i \theta}}{\beta +\epsilon} \left[
(\alpha_1 + \epsilon + 2) R^2 e^{2 i \theta} - [\alpha_1 -(\epsilon + 2)] \right]
\\ 
& \, \, \quad  + 2 \, \frac{16 - \beta + 3 \epsilon}{\beta +\epsilon} \, 
R^2 e^{2 i \theta} + 1 \\
&=  R^4  \left\{ 
e^{4 i \theta} + \frac{4 e^{i \theta}}{R(\beta +\epsilon)} \left[
(\alpha_1 + \epsilon + 2) e^{2 i \theta} - \frac{1}{R^2} \, [\alpha_1 -(\epsilon + 2)] \right] \right. \\
& \quad \, \,  \qquad \left.
+ \frac{2}{R^2} \, \frac{16 - \beta + 3 \epsilon}{\beta +\epsilon} \, 
e^{2 i \theta} + \frac{1}{R^4} \right\} \, \, .
\end{align*}
We have the following parametrisations of the image of the boundary 
of $D$ under $q$:
\begin{align*}
& \left(\, 
[0,R] \rightarrow {\mathbb{C}} , x \mapsto q(x)
\,\right) \, \, , \\
& \left(\, 
[0,\pi/2] \rightarrow {\mathbb{C}} , \theta \mapsto q\left(R e^{i \theta}\right)
\,\right) \, \, , \\
& \left(\, 
[0,R] \rightarrow {\mathbb{C}} , y \mapsto q(i (R - y))
\,\right) \, \, .
\end{align*}
Hence it follows, according to Lemmas~\ref{basiclemma1},~\ref{basiclemma2} and for sufficiently 
large $R > 0$, that these parametrisations, starting from the point $(0,1)$, 
through the open $4$-th quadrant, into the open $1$-st quadrant, through 
the open $2$-nd and $3$-rd quadrants, back into the open $4$-th quadrant, 
crossing the imaginary axis into the open $3$-rd quadrant, crossing the 
imaginary axis again into the open $4$-th quadrant, before reaching the point
$(0,1)$ again. Thus the increase in argument of $q$ around the boundary of
$D$ is $2 \pi$, and the open $1$-st quadrant contains precisely $1$ root of $q$.
\end{proof}

\begin{theorem} \label{maintheorem}
If $m, k \in {\mathbb{N}}$ are such that
\begin{equation} \label{newinstabilitycondition1}
m \geqslant 2k+1 + \sqrt{6 k^2 + 6 k + 1} \, \, ,
\end{equation}
and $a \in [0,M)$ is such 
\begin{equation} \label{newinstabilitycondition2}
1 > \frac{a}{M} > \frac{2 \sqrt{6}}{5} \, \frac{\sqrt{1 + \frac{1}{m} \left[2k+1 + \frac{k(k+1)}{m}\right]}}{
1 + \frac{2/5}{m} \left[2k+1 + \frac{k(k+1)}{m}\right]
}  \, \, ,
\end{equation}
then
the open interval 
\begin{equation*}
I := \left(\sqrt{\frac{1}{6} \, l(l+1)} \, \, , \, \frac{\frac{m a}{M}}{2\left(1 + \sqrt{1 - \frac{a^2}{M^2}}\right)} 
\right) 
\end{equation*}
is non-empty, and for every $\mu M \in I$,   
the open first quadrant contains precisely $1$ root of $q$. We note that 
\begin{equation*}
\frac{2 \sqrt{6}}{5} \approx 0.979796 \, \, .
\end{equation*}
\end{theorem}

\begin{proof}
According to Theorem~\ref{instabilitytheorem}, if 
\begin{align*}
\alpha_1 = \frac{ma}{\mu M^2} > 2 + \epsilon \, \, \wedge \, \, 
0 < \beta = 6 - 
\frac{l(l+1)}{\mu^2 M^2} < 8 + \epsilon \, \, ,
\end{align*}
then 
the open first quadrant contains precisely $1$ root of $q$.
We note the equivalence of the following inequalities
\begin{align*}
&  0 < 6 - \frac{l(l+1)}{\mu^2 M^2} < 8 + \epsilon \, \, \Leftrightarrow \, \, 
0 >  - 6 +  \frac{l(l+1)}{\mu^2 M^2}  > - (8 + \epsilon) \, \, , \\
&  6 >  \frac{l(l+1)}{\mu^2 M^2} > - (2 + \epsilon) \, \,\Leftrightarrow \, \,
6 >  \frac{l(l+1)}{\mu^2 M^2} \, \, , \, \, \mu^2 M^2 > \frac{1}{6} \, l(l+1) 
\, \, , 
\end{align*}
where we used that $l \geqslant 0$,
as well as
\begin{align*}
& 2 + \epsilon < \frac{m a}{\mu M^2} \, \, \Leftrightarrow \, \, 
\mu M  < \frac{m a}{(2 + \epsilon) M}  \, \, \Leftrightarrow \, \,
\mu^2 M^2  < \frac{m^2 a^2}{(2 + \epsilon)^2 M^2} \wedge m \geqslant 0
.
\end{align*}
Hence these $2$ inequalities can be joined to 
\begin{equation*}
\frac{1}{6} \, l(l+1) < \mu^2 M^2 <  \frac{(a/M)^2}{(2 + \epsilon)^2} \, m^2
\wedge m \geqslant 0 \, \, .
\end{equation*}
Since $l = m + k $, where $k \in {\mathbb{N}}$, and 
if  
\begin{equation*}
m \geqslant 2k+1 + \sqrt{6 k^2 + 6 k + 1} \, \, , 
\end{equation*}
we note the equivalence 
of the following inequalities
\begin{align*}
& \frac{1}{6} \, l(l+1) < \frac{(a/M)^2}{(2 + \epsilon)^2} \, m^2 \\
&  \Leftrightarrow \frac{l(l+1)}{6 m^2} < \frac{(a/M)^2}{(2 + \epsilon)^2} = 
\frac{\frac{a^2}{M^2}}{4 \left(1 + \sqrt{1 - \frac{a^2}{M^2}}
\,  \right)^2} \\
& \Leftrightarrow \frac{2l(l+1)}{3m^2} \left(1 + \sqrt{1 - \frac{a^2}{M^2}} 
\,  \right)^2 < \frac{a^2}{M^2} \\
& \Leftrightarrow \frac{2l(l+1)}{3m^2} \left(2  + 2 \sqrt{1 - \frac{a^2}{M^2}} - \frac{a^2}{M^2}
\,  \right) < \frac{a^2}{M^2} \\
& \Leftrightarrow 
2  + 2 \sqrt{1 - \frac{a^2}{M^2}} - \frac{a^2}{M^2}
< \frac{3m^2}{2l(l+1)} \, \frac{a^2}{M^2} \\
& \Leftrightarrow 
2 \sqrt{1 - \frac{a^2}{M^2}}
< \left[1 + \frac{3m^2}{2l(l+1)}\right] \frac{a^2}{M^2} - 2 \\
& \Leftrightarrow 
4 \left(1 - \frac{a^2}{M^2}\right)
< \left\{\left[1 + \frac{3m^2}{2l(l+1)}\right] \frac{a^2}{M^2} - 2 \right\}^2 \\
& \Leftrightarrow - 4 \, \frac{a^2}{M^2} < \left[1 + \frac{3m^2}{2l(l+1)}\right]^2 \frac{a^4}{M^4} - 4 \left[1 + \frac{3m^2}{2l(l+1)}\right] \frac{a^2}{M^2} \\
& \Leftrightarrow 0 < \left[1 + \frac{3m^2}{2l(l+1)}\right]^2 \frac{a^4}{M^4} - \frac{6 m^2}{l(l+1)} \, \frac{a^2}{M^2} \\
& \Leftrightarrow 0 < \left\{ \left[1 + \frac{3m^2}{2l(l+1)}\right]^2 \frac{a^2}{M^2} - \frac{6 m^2}{l(l+1)} \right\} \frac{a^2}{M^2} \\
& \Leftrightarrow 0 < \left[1 + \frac{3m^2}{2l(l+1)}\right]^2 \frac{a^2}{M^2} - \frac{6 m^2}{l(l+1)} \\
& \Leftrightarrow \frac{6 m^2}{l(l+1)} < \left[1 + \frac{3m^2}{2l(l+1)}\right]^2 \frac{a^2}{M^2}  \\
& \Leftrightarrow \frac{6 m^2}{l(l+1)} < \left[\frac{3m^2 + 2l(l+1)}{2l(l+1)}\right]^2 \frac{a^2}{M^2} \\
& \Leftrightarrow \frac{6 m^2}{l(l+1)} \left[\frac{2l(l+1)}{3m^2 + 2l(l+1)}\right]^2 <  \frac{a^2}{M^2} \\
&  \Leftrightarrow \frac{24 m^2l(l+1)}{[3m^2 + 2l(l+1)]^2} <  \frac{a^2}{M^2} \\
&  \Leftrightarrow \frac{24 m^2 (m+k)(m+k+1)}{[3m^2 + 2(m+k)(m+k+1)]^2} <  \frac{a^2}{M^2} \\
&  \Leftrightarrow \frac{24 (1+\frac{k}{m})(1+\frac{k+1}{m})}{[3 + 2(1+\frac{k}{m})(1+\frac{k+1}{m})]^2} <  \frac{a^2}{M^2} \\
&  \Leftrightarrow \frac{24}{25} \, \frac{1 + \frac{1}{m} \left[2k+1 + \frac{k(k+1)}{m}\right]}{\left\{1 + \frac{2/5}{m} \left[2k+1 + \frac{k(k+1)}{m}\right]
\right\}^2} <  \frac{a^2}{M^2} \, \, ,
\end{align*}
where we note for the validity of these equivalences that 
\begin{align*}
& \frac{6 m^2}{3m^2 + 2l(l+1)} \geqslant 1 
\Leftrightarrow 6 m^2 \geqslant 3m^2 + 2l(l+1) \\
& \Leftrightarrow 3 m^2 \geqslant 2 l(l+1) = 2 (m+k)(m+k+1) \Leftrightarrow
m^2 \geqslant 2(2k+1) m + 2 k(k+1) \\
& \Leftrightarrow m^2 - 2(2k+1) m - 2 k(k+1) \geqslant 0 \Leftrightarrow
[m - (2k+1)]^2 - (2k+1)^2 - 2 k(k+1) \geqslant 0 \\
& \Leftrightarrow  
[m - (2k+1)]^2 \geqslant 6 k^2 + 6 k + 1 \Leftrightarrow 
m \geqslant 2k+1 + \sqrt{6 k^2 + 6 k + 1} \, \, ,
\end{align*}
and, if 
\begin{equation*}
m \geqslant 2k+1 + \sqrt{6 k^2 + 6 k + 1} \, \, , 
\end{equation*}
then 
\begin{align*}
& \frac{a^2}{M^2} > \frac{24 m^2l(l+1)}{[3m^2 + 2l(l+1)]^2}  \Leftrightarrow 
\frac{a^2}{M^2} > \frac{6 m^2}{3m^2 + 2l(l+1)} \, \frac{4 l(l+1)}{3m^2 + 2l(l+1)} \\
& \Rightarrow 
\frac{a^2}{M^2} > \frac{4 l(l+1)}{3m^2 + 2l(l+1)}
\Leftrightarrow \left[1 + \frac{3m^2}{2l(l+1)}\right] \frac{a^2}{M^2} - 2 > 0 
\, \, .
\end{align*}
\end{proof}


\section{Discussion of the Results}

The present paper is a follow-up of our previous paper that derives a slightly simplified model equation
for the Klein-Gordon equation, describing the propagation of a scalar field of mass $\mu$
in the background of a rotating black hole and, among others things, supports 
the instability of the field down to 
$a/M \approx 0.97$. The latter result was derived numerically. This paper
gives corresponding rigorous results, supporting instability of the field down to 
$a/M \approx 0.979796$. This result supports claims 
of previous rigorous as well as analytical and numerical investigations that show 
instability of 
the massive Klein-Gordon field for $a/M$ 
extremely close to $1$. 
\newline
\linebreak 
From here, mathematical investigation could proceed 
in $2$ directions. First, it might be possible to use the 
model for the proof of the instability of the massive 
Klein-Gordon equation in a Kerr background, using a perturbative 
approach, in this way complementing the result of Shlapentokh-Rothman, 
(\cite{shlapen}, 2014). Another direction consists in further simplification 
of the model in order to find the mathematical root 
of the instability as well as an abstraction to a larger 
class of equations that includes the massive Klein-Gordon 
equation on a Kerr background. It is tempting to assume that the 
instability is due to particular commutation properties
of the operators $A$ and $B$ governing the evolution equation, 
(\ref{abstractwaveequation}).

\section*{Acknowledgments} 
H.B. is thankful for the hospitality and support
of the `Department of Gravitation and Mathematical Physics',
(ICN, Miguel Alcubierre), 
Universidad Nacional Autonoma de 
Mexico, Mexico City, Mexico and 
the `Division Theoretical Astrophysics' (L. Rezzolla) of the Institute of Theoretical 
Physics at the Goethe University Frankfurt, Germany. 
This work was supported in part by CONACyT
grants 82787 and 167335, DGAPA-UNAM through grant IN115311, SNI-M\'exico, and 
the ERC Synergy Grant “BlackHoleCam: Imaging the Event Horizon of Black Holes” (Grant No. 610058). M.M. acknowledges DGAPA-UNAM for a postdoctoral grant.

\end{document}